\documentclass[12pt]{article}
\usepackage{amsmath, amsfonts, amssymb, amsthm, graphicx, geometry, arydshln, umoline, subfig, newtxtext, newtxmath, color, comment, soul, enumerate, bm,float,multirow,booktabs,varwidth}
\usepackage{tikz}
\usetikzlibrary{matrix, positioning,fit,calc,arrows.meta}
\usepackage{appendix}
\usepackage[colorlinks=true,citecolor=blue]{hyperref}
\usepackage{natbib}
\usepackage{booktabs}
\usepackage[table]{xcolor}
\def\BibTeX{{\rm B\kern-.05em{\sc i\kern-.025em b}\kern-.08em
    T\kern-.1667em\lower.7ex\hbox{E}\kern-.125emX}}
\usepackage{makecell}
\usepackage{bbm}

\numberwithin{equation}{section}

\theoremstyle{definition}
\newtheorem{theorem}{Theorem}

\newtheorem{corollary}{Corollary}

\newtheorem{definition}{Definition}
\newtheorem{example}{Example}

\newtheorem{lemma}{Lemma}

\newtheorem{proposition}{Proposition}

\begin{document}

\title{\fontsize{20pt}{1cm}\selectfont{Recursive Borda Aggregation\thanks{This paper supersedes a  previous version circulated as ``Iterative Elimination of Borda Losers: Axiomatizations of the Baldwin and Nanson Rules''.
We thank Satoru Fujishige, Takashi Kurihara, Vincent Merlin, Masashi Umezawa and participants at the spring meeting of the Operations Research Society of Japan 2026 and SSCW 2026 for their helpful comments.
Satoshi Nakada acknowledges the financial support by JSPS KAKENHI Grant Number 25K16606 and 25K00618.
All the remaining errors are our own.}
}}
\author{
Leo Goto\thanks{Undergraduate School of Management, Department of Business Economics, Tokyo University of Science, 1-11-2, Fujimi, Chiyoda-ku, Tokyo, 102-0071, Japan. Email: leogotodeb@gmail.com}
\and
Satoshi Nakada\thanks{School of Management, Department of Business Economics, Tokyo University of Science, 1-11-2, Fujimi, Chiyoda-ku, Tokyo, 102-0071, Japan. Email: snakada@rs.tus.ac.jp}
}
\date{\today}
\maketitle

\begin{abstract}
We study a new class of voting rules, \textit{recursive Borda rules}, in which the Borda criterion is
applied recursively through successive elimination decisions rather than in a single step.
Our first result provides an axiomatic foundation for a regular subclass of recursive Borda rules.
We show that recursive reformulations of Young's (\citeyear{Young1974}) axioms characterize a regular subclass of recursive Borda rules.
Our second result establishes an exact connection between recursive Borda rules and the Condorcet criterion.
We characterize precisely when a recursive Borda rule satisfies \textit{Condorcet Consistency}. 
As a consequence, every regular recursive Borda rule satisfies \textit{Condorcet Consistency}, and \textit{Condorcet Consistency} identifies a sharp upper bound on admissible elimination regions.
Within this class, the Baldwin rule, the strict Nanson rule, and the Nanson rule emerge as the three canonical procedures, and we provide unified axiomatic characterizations of all three.
Taken together, our results extend Young's axiomatic theory of the Borda rule from one-shot positional aggregation to recursive aggregation and derive \textit{Condorcet Consistency} rather than postulating it as an independent axiom.

\end{abstract}

\section{Introduction}
\subsection{Motivation}
Voting theory has long been shaped by two influential approaches to preference aggregation originating in the works of Borda~\citep{Borda1781} and Condorcet~\citep{Condorcet1785}.\footnote{Borda introduced what is now known as the \emph{Borda count} as a positional method based on complete rankings, whereas Condorcet advocated selecting an alternative that defeats every opponent in pairwise majority comparisons whenever such an alternative exists. 
The terminology \textit{Condorcet criterion} was later popularized by \citet{Black1958}.
}
The Borda rule evaluates alternatives according to their aggregate positions in voters' complete rankings, whereas the Condorcet criterion evaluates an alternative through pairwise majority comparisons.
Although both approaches seek socially desirable outcomes, they rely on fundamentally different principles of aggregation.

The distinction between these approaches is not merely procedural but also informational.
As emphasized by \citet{Saari1994,Saari1995,Saari2000,Saari2001}, positional and pairwise methods aggregate different aspects of voters' preferences.
Pairwise-majority methods aggregate only binary majority comparisons, whereas positional methods exploit information contained in complete rankings.
Consequently, the two approaches need not agree even when applied to the same preference profile.

The informational distinction between the two traditions also has an important axiomatic counterpart.
\citet{Young1974} characterized the Borda rule by four axioms: \textit{Neutrality}, \textit{Consistency}, \textit{Faithfulness}, and \textit{Cancellation}.\footnote{\citet{Young1974} characterizes the Borda \emph{choice correspondence}, whereas \citet{NitzanRubinstein1981} provide the corresponding characterization of the complete social ranking.}
Among these axioms, \emph{Consistency} occupies a central position.
It captures the additive structure of Borda aggregation by requiring social comparisons supported separately by two disjoint electorates to be preserved when those electorates are combined.\footnote{Closely related reinforcement principles have long played a central role in the literature on social choice correspondences; see, for example, \citet{YoungLevenglick1978}.}
In this sense, \textit{Consistency} expresses the aggregation principle underlying the Borda rule.

Reconciling the Borda principle with the Condorcet criterion has long been a central problem in social choice.\footnote{The literature on Condorcet-consistent voting rules itself is extensive; classical examples include the Black, Copeland, Dodgson, and Schulze rules \citep{Black1958,Copeland1951,Dodgson1876,Schulze1997,Schulze2025}.
}
There are, of course, many possible approaches to this problem. 
We focus on an approach that retains Borda evaluation itself and changes only how it is applied. 
Whereas the Borda rule uses a single
score calculation to determine the entire social ranking, a recursive procedure uses the current Borda scores to identify the alternatives to be placed in the current bottom tier, removes them, and then reapplies the same criterion to the remaining alternatives. 
This is a conservative extension of the Borda principle: the information used
and the local criterion of evaluation remain unchanged, while the complete ranking is constructed through a sequence of elimination decisions.

This approach is particularly well motivated by \citet{Smith1973}, who showed that, among elimination procedures based on scoring rules, the Borda score occupies a unique position in guaranteeing
\textit{Condorcet Consistency}.\footnote{The Borda score possesses several other distinctive properties among scoring rules. 
For example, \citet{OkamotoSakai2019} show that it is the unique scoring rule satisfying a natural variant of the Condorcet loser criterion.}
Thus, although recursive elimination is not the only possible route from the Borda principle to the Condorcet criterion, recursive use of the Borda score provides a distinguished route within the framework of positional aggregation.

The Baldwin~\citep{Baldwin1926} and Nanson~\citep{Nanson1882} rules provide perhaps the most natural realization of this recursive perspective.\footnote{Unlike most Condorcet-consistent voting rules, the Baldwin and Nanson rules have also been implemented in practice, most notably in elections at the University of Melbourne and in municipal elections in Marquette, Michigan; see \citet{Gruber2025}. Gruber further reports that the complexity of the counting procedure eventually led to the repeal of the Marquette ordinance in 1933.}
A natural intermediate procedure, which we call the strict Nanson rule, eliminates exactly those alternatives whose Borda scores are strictly below the average.
Like the Borda rule, all three rules evaluate the alternatives exclusively through Borda scores computed from voters' complete rankings.
Their distinguishing feature is that these Borda evaluations are repeatedly applied after eliminating poorly performing alternatives.
Thus, the principle of social evaluation continues to be based on Borda scores, while the aggregation procedure itself becomes recursive.
An important consequence is that the three rules satisfy the Condorcet criterion.

Although the three rules differ only in how they eliminate alternatives, they need not select the same winner.
As noted by \citet{Niou1987}, a normative understanding of these differences has remained elusive.
These three rules thus raise two questions.
First, which axioms characterize recursive Borda aggregation itself?
Second, once recursive Borda aggregation has been identified, which additional principles distinguish its canonical members?
In particular, how should Young's \textit{Consistency} be reformulated for recursive Borda aggregation so as to retain the Borda principle while recovering the Condorcet criterion?

Young's \textit{Consistency} governs the complete social ranking, whereas a recursive rule constructs that ranking through successive elimination decisions.
As we show, no recursive Borda rule satisfies full \textit{Consistency}.
This motivates a recursive reformulation of Young's aggregation principle.
Instead of requiring consistency of the final ranking, we require consistency only of the current elimination decision and then apply the same rule recursively to the reduced problem.
These two ideas are formalized by \emph{Bottom Consistency} and \emph{Bottom Independence}, respectively.

Our contribution is threefold.
First, we introduce the class of recursive Borda rules and provide its first axiomatic foundation.
Rather than postulating recursive use of Borda scores, we derive regular recursive Borda rules from recursive reformulations of Young's axioms.
Second, we establish a sharp connection between recursive Borda aggregation and the Condorcet criterion.
We characterize exactly when a recursive Borda rule satisfies \textit{Condorcet Consistency}.
As a consequence, every regular recursive Borda rule satisfies \textit{Condorcet Consistency}, and \textit{Condorcet Consistency} imposes a simple upper bound on admissible elimination regions.
The Baldwin rule, the strict Nanson rule, and the Nanson rule emerge as the three canonical procedures within this class.
Third, we provide unified axiomatic characterizations of these three
canonical rules.
The Baldwin rule is identified by a stronger recursive consistency requirement, whereas the strict Nanson and Nanson rules are distinguished by the treatment of surviving and zero-score alternatives.

The remaining axioms are adapted in the same recursive spirit.
Young's \textit{Faithfulness} is replaced by \emph{Weak Faithfulness},
requiring only that a single voter's most preferred alternative not be eliminated in the first round.
Boundary conditions on zero-score alternatives are captured by \emph{Dummy Retention} and \emph{Dummy Exclusion}.
Together these axioms yield a unified framework for recursive Borda
aggregation.

Our results reveal that the Baldwin, strict Nanson, and Nanson rules share a common recursive foundation inherited from Young's characterization of the Borda rule and differ only in a small number of recursive principles.
The resulting axiomatic foundations also uncover a sharp distinction with respect to intermediate Condorcet winners and losers \citep{BarberaBossert2025}.
To the best of our knowledge, this is the first unified axiomatic treatment of the recursive Borda aggregation.

The rest of the paper is organized as follows.
Section~\ref{sec:model} introduces the model and notation.
Section~\ref{sec:recursive-borda} introduces recursive Borda aggregation rules and provides their axiomatic foundation.
Section~\ref{sec:condorcet-safe-interval} establishes the connection between recursive Borda aggregation and Condorcet Consistency and identifies the Baldwin, strict Nanson, and Nanson rules as the three canonical procedures.
Section~\ref{sec:canonical-axiomatizations} presents unified direct axiomatic characterizations of these three rules.
Section~\ref{sec:discussion} concludes.

\subsection{Related literature}

Our paper contributes to the literature on the axiomatic foundations of positional voting, particularly the Borda rule and its recursive extensions.
Young's seminal characterization of the Borda rule \citep{Young1974}, later extended to complete social rankings by \citet{NitzanRubinstein1981}, established \textit{Consistency} as the fundamental aggregation principle underlying positional voting.
This characterization has inspired a substantial literature on the axiomatic foundations of Borda-type rules.
Alternative proofs are given by \citet{HanssonSahlquist1976} and \citet{mihara2017}, while \citet{barbera2023} characterize both the Borda rule and majority rule in an opinion aggregation framework.
More recently, \citet{barbera2026} characterize the Black rule, which selects the Condorcet winner whenever one exists and otherwise chooses the Borda winner.

Our work extends this line of research from one-shot positional aggregation to recursive aggregation.
Complementing Young's axiomatic approach, \citet{Smith1973} shows that the Borda score is the unique scoring method yielding \textit{Condorcet Consistency} within the class of score-based elimination procedures.
Whereas Smith identifies the distinguished role of the Borda score from an algorithmic perspective, we provide the corresponding axiomatic foundation by identifying recursive reformulations of Young's framework.
Relatedly, \citet{YoungLevenglick1978} reveal the delicate relationship between \textit{Consistency} and the Condorcet principle, while \citet{OkamotoSakai2019} identify another distinctive property of the Borda rule through a natural variant of the Condorcet loser criterion.

Recent work has also investigated axiomatic foundations for recursive elimination rules.
\citet{FreemanBrillConitzer2014} characterize the Baldwin rule using \textit{Neutrality}, \textit{Anonymity}, \textit{Strong Bottom Consistency}, \textit{Bottom Independence}, \textit{Continuity at the Bottom}, and \textit{Condorcet Consistency}.\footnote{\citet{FreemanBrillConitzer2014} consider social preference functions based on the parallel-universes tie-breaking framework (PUT).}
\citet{BrandtDongPeters2025} characterize the Nanson rule using a strengthened form of \textit{Condorcet Consistency} in the three-alternative case.
Our work differs in three respects.
First, we introduce and axiomatize the broader class of recursive Borda aggregation rules before identifying individual voting rules.
Second, our characterizations apply to an arbitrary number of alternatives, thereby extending the scope of \citet{BrandtDongPeters2025}.
Third, unlike both \citet{FreemanBrillConitzer2014} and \citet{BrandtDongPeters2025}, we do not impose any form of \textit{Condorcet Consistency} axiomatically.
Instead, \textit{Condorcet Consistency} is derived from recursive reformulations of Young's framework.
Both the Baldwin and Nanson rules can also be viewed as margin-matrix rules, and \citet{ding2025characterizations} characterize this broader class.

From a computational perspective, \citet{NarodytskaWalshXia2011} show that manipulation of the Baldwin and Nanson rules is NP-hard, while \citet{FavardinLepelley2006} find them less susceptible to manipulation than positional scoring rules in elections with three alternatives.
Moreover, \citet{BoehmerBrederekPeters2026} show that determining whether a candidate can win or be ranked in a given position under recursive scoring rules, including the Baldwin rule, is generally NP-hard.
More broadly, \citet{AndoSugimotoSukegawa2026} develop recursive heuristics for the linear ordering problem, showing that a recursive Borda method satisfies the Condorcet criterion while remaining computationally efficient.

Finally, our proofs build on the linear-algebraic approach initiated by \citet{Young1974}.
In particular, \citet{Zwicker1991} developed the cycle--cocycle decomposition of majority-margin matrices, which plays a central role in our analysis of recursive Borda aggregation.
Related linear-algebraic techniques have also been employed by \citet{DuddyPigginsZwicker2016} in characterizing the mean rule and by \citet{BrandlPeters2019} in characterizing the Borda mean rule.

\section{The model and notation}
\label{sec:model}

\subsection{Preferences and profiles}
\label{subsec:preferences}

Let $\mathcal N=\{1,2,\ldots\}$ be a countably infinite set of potential voters.
Let $\mathcal U$ denote the collection of all finite and nonempty subsets of $\mathcal N$.
An element $N\in\mathcal U$ is called an electorate, and we write $n=|N|$ for its cardinality.
Let $A=\{a_1,\ldots,a_m\}$ with $m\geq3$ be the universal set of alternatives.
Let $\mathcal B=\{B\subseteq A:B\neq\emptyset\}$ denote the collection of all nonempty feasible sets.
For every $B\in\mathcal B$, let $\mathcal L(B)$ denote the set of strict linear orders  on $B$ (i.e., complete, transitive, and irreflexive binary relations).

For $\succ_i\,\in\mathcal L(B)$ and $a\in B$, define the rank of $a$ under $\succ_i$ by
\[
r_B(\succ_i,a)=
\left|
\{b\in B:b\succ_i a\}
\right|
+1.
\]
Thus, $r_B(\succ_i,a)=1$ if and only if $a$ is the most preferred alternative under $\succ_i$.
For every nonempty $C\subseteq B$ and every $\succ_i\,\in\mathcal L(B)$, let $\succ_i\!\!|_C\in\mathcal L(C)$ denote the restriction of $\succ_i$ to $C$, where the restriction is defined by $a\succ_i\!\!|_C b \Longleftrightarrow a\succ_i b$ for all $a,b\in C$.
For every electorate $N\in\mathcal U$ and every feasible set $B\in\mathcal B$, define $\mathcal P^N(B)=\prod_{i\in N}\mathcal L(B)$.
An element $\succ^N=(\succ_i^N)_{i\in N}\in\mathcal P^N(B)$ is called a preference profile on $B$ for electorate $N$.
For every nonempty $C\subseteq B$, define the restriction of $\succ^N$ to $C$ by $\succ^N\!\!|_C
=(\succ_i^N\!\!|_C)_{i\in N}\in \mathcal P^N(C)$.

Let $N,N'\in\mathcal U$ be disjoint electorates.
For profiles $\succ^N\in\mathcal P^N(B)$ and $\succ^{N'}\in\mathcal P^{N'}(B)$, their union is the profile $\succ^N\cup\succ^{N'}\in \mathcal P^{N\cup N'}(B)$ defined by
\[
(\succ^N\cup\succ^{N'})_i
=
\begin{cases}
\succ_i^N,
&
i\in N,
\\
\succ_i^{N'},
&
i\in N'.
\end{cases}
\]
For every $k\in\mathbb N$, let $k*\!\succ^N$ denote a $k$-fold replication of $\succ^N$.
That is, $k*\!\succ^N$ is the union of $k$ pairwise disjoint copies of $\succ^N$.
Define $\mathcal P(B)=\bigcup_{N\in\mathcal U}\mathcal P^N(B)$ and $\mathcal P=\bigcup_{B\in\mathcal B}\mathcal P(B)$.
The electorate and the feasible set are regarded as part of the description of a profile.
When no confusion can arise, we suppress the electorate superscript and write simply $\succ$.

\subsection{Social rankings and voting rules}
\label{subsec:social-rankings}

For every feasible set $B\in\mathcal B$, let $\mathcal R(B)$ denote the set of weak orders on $B$ (i.e., complete and transitive binary relations).
Let $\mathcal R=\bigcup_{B\in\mathcal B}\mathcal R(B)$.
For a weak order $R\in\mathcal R(B)$, let $P_R$ and $I_R$ denote its asymmetric and symmetric parts, respectively.
Thus, $aP_Rb$ holds if $aRb$ and not $bRa$, whereas $aI_Rb$ holds if both $aRb$ and $bRa$.
For every nonempty $C\subseteq B$, let $R|_C$ denote the restriction of $R$ to $C$.

A voting rule is a function $F:\mathcal P\longrightarrow\mathcal R$
such that $F(\succ)\in\mathcal R(B)$ for every $\succ\,\in\mathcal P(B)$.
For a profile $\succ\,\in\mathcal P(B)$, we write $R_{F(\succ)}$ for the weak order selected by $F$.
We write $aR_{F(\succ)}b$, $aP_{F(\succ)}b$, and $aI_{F(\succ)}b$ for the corresponding weak, strict, and indifferent social comparisons.

Define the top tier of $F(\succ)$ by
\[
\operatorname{top}(F(\succ))
=
\left\{
a\in B:
aR_{F(\succ)}b
\text{ for every }b\in B
\right\},
\]
and define the bottom tier of $F(\succ)$ by
\[
\operatorname{bottom}(F(\succ))
=
\left\{
a\in B:
bR_{F(\succ)}a
\text{ for every }b\in B
\right\}.
\]
Because $B$ is finite and $F(\succ)$ is a weak order, both tiers are nonempty.
If all alternatives are socially indifferent, we write $F(\succ)=(B)$.
Thus, $F(\succ)=(B)$ if and only if $\operatorname{top}(F(\succ))=\operatorname{bottom}(F(\succ))=B$.

Let $\Pi(A)$ denote the set of permutations of $A$.
For $\sigma\in\Pi(A)$ and $B\in\mathcal B$, define $\sigma(B)=\{\sigma(a):a\in B\}$.
For $\succ_i\,\in\mathcal L(B)$, define the relabeled preference $\sigma(\succ_i)\in\mathcal L(\sigma(B))$ by $\sigma(a)\,\sigma(\succ_i)\,\sigma(b)\Longleftrightarrow a\succ_i b$ for every $a,b\in B$.
The relabeled profile $\sigma(\succ^N)$ is obtained by applying $\sigma$ to every individual preference relation.
Similarly, for $R\in\mathcal R(B)$, define $\sigma(R)\in\mathcal R(\sigma(B))$ by $\sigma(a)\,\sigma(R)\,\sigma(b) \Longleftrightarrow aRb$ for every $a,b\in B$.

\subsection{Majority margins and Borda scores}
\label{subsec:borda-scores}

For a profile $\succ^N\in\mathcal P^N(B)$, define its majority-margin matrix $T(\succ^N)\in\mathbb Z^{B\times B}$ by $[T(\succ^N)]_{ab}=\left|\{i\in N:a\succ_i^N b\}\right|-\left|\{i\in N:b\succ_i^N a\}\right|$ for every $a,b\in B$.
By definition, each majority-margin matrix is skew-symmetric: $T(\succ^N)^\top=-T(\succ^N)$.
Let $0_B$ denote the zero matrix in $\mathbb Z^{B\times B}$.
For every nonempty $C\subseteq B$, let $T(\succ^N)|_C$ denote the principal submatrix of $T(\succ^N)$ indexed by $C$.
Restriction of a profile corresponds to restriction of its majority-margin matrix: $T(\succ^N\!\!|_C)=T(\succ^N)|_C$.
Note that majority-margin matrices are additive across disjoint electorates: $T(\succ^N\cup\succ^{N'})=T(\succ^N)+T(\succ^{N'})$ and homogeneous under replication: $T(k*\!\succ^N)=kT(\succ^N)$.
For $\sigma\in\Pi(A)$ and a matrix $T\in\mathbb R^{B\times B}$, define the relabeled matrix $\sigma(T)\in\mathbb R^{\sigma(B)\times\sigma(B)}$ by $[\sigma(T)]_{\sigma(a)\sigma(b)}=T_{ab}$ for every $a,b\in B$.
The majority-margin matrix is compatible with relabeling: $T(\sigma(\succ))=\sigma(T(\succ))$.

A voting rule $F$ is a \emph{margin matrix rule} if, for every $B\in\mathcal B$ and every $\succ,\succ'\,\in\mathcal P(B)$, if $T(\succ)=T(\succ')$, then $F(\succ)=F(\succ')$.
For every profile $\succ\,\in\mathcal P(B)$ and every $a\in B$, define the centered Borda score of $a$ by
\[
\beta_a^B(\succ)
=
\sum_{b\in B\setminus\{a\}}
[T(\succ)]_{ab}.
\]
Let $\beta^B(\succ)=\left(\beta_a^B(\succ)\right)_{a\in B}\in\mathbb Z^B$ denote the centered Borda-score vector.
When the feasible set is clear from the context, we write $\beta_a(\succ)$ and $\beta(\succ)$.
Note that the centered Borda scores sum to zero: $\sum_{a\in B}\beta_a^B(\succ)=0$.
Let $\bm 0_B$ denote the zero vector in $\mathbb R^B$.
Note also that centered Borda-score vectors are additive across disjoint electorates:$\beta^B(\succ^N\cup\succ^{N'})=\beta^B(\succ^N)+\beta^B(\succ^{N'})$ and homogeneous under replication: $\beta^B(k*\!\succ^N)=k\beta^B(\succ^N)$.
For a vector $v=(v_a)_{a\in B}\in\mathbb R^B$ and a permutation $\sigma\in\Pi(A)$, define $\sigma(v)\in\mathbb R^{\sigma(B)}$ by $[\sigma(v)]_{\sigma(a)}=v_a$ for every $a\in B$.
Centered Borda-score vectors are compatible with relabeling: $\beta^{\sigma(B)}(\sigma(\succ))=\sigma\bigl(\beta^B(\succ)\bigr)$.
For every feasible set $B\in\mathcal B$, let
\[
\mathcal S(B)
=
\left\{
s\in\mathbb Q^B:
\sum_{a\in B}s_a=0
\right\}
\]
denote the space of rational centered score vectors on $B$.

To relate the centered score to the standard positional Borda score, let $q=|B|$ and define
\[
b_a^B(\succ^N)
=
\sum_{i\in N}
\left(
q-r_B(\succ_i^N,a)
\right).
\]
Then, the centered and positional Borda scores satisfy
\[
\beta_a^B(\succ^N)
=
2b_a^B(\succ^N)
-
n(q-1).
\]
Hence, the centered Borda score is an affine transformation of the standard positional Borda score and induces the same ranking.
The average positional Borda score is $n(q-1)/2$.
Consequently, $\beta_a^B(\succ)\leq0$ if and only if the positional Borda score of $a$ is weakly below the average.
Similarly, $\beta_a^B(\succ)<0$ if and only if the positional Borda score of $a$ is strictly below the average.
The \textit{Borda rule}, denoted by $F^{\mathrm{Borda}}$, ranks alternatives according to their centered Borda scores:
\[
aR_{F^{\mathrm{Borda}}(\succ)}b \Longleftrightarrow \beta_a^B(\succ)
\geq
\beta_b^B(\succ)
\]
for every $a,b\in B$.

\subsection{Young's axioms}
\label{subsec:young-benchmark}

\citet{Young1974} characterizes the Borda choice correspondence, and \citet{NitzanRubinstein1981} provide the corresponding characterization of the complete Borda ranking.\footnote{\citet{Young1974} considers the choice correspondence $C(\succ)=\operatorname{top}(F^{\mathrm{Borda}}(\succ))$.}
Applied to every feasible set, the benchmark axioms are the following.

\medskip
\noindent
\textbf{Neutrality.}
For every $B\in\mathcal B$, every $\succ\,\in\mathcal P(B)$, and every $\sigma\in\Pi(A)$, $F(\sigma(\succ))=\sigma(F(\succ))$.

\medskip
\noindent
\textbf{Consistency.}
For every $B\in\mathcal B$, every pair of disjoint electorates $N,N'\in\mathcal U$, every $\succ^N\in\mathcal P^N(B)$ and $\succ^{N'}\in\mathcal P^{N'}(B)$, and every $a,b\in B$, if $aR_{F(\succ^N)}b$ and $aR_{F(\succ^{N'})}b$, then $aR_{F(\succ^N\cup\succ^{N'})}b$.
If, in addition, $aP_{F(\succ^N)}b$ or $aP_{F(\succ^{N'})}b$, then $aP_{F(\succ^N\cup\succ^{N'})}b$.

\medskip
\noindent
\textbf{Faithfulness.}
For every $B\in\mathcal B$ and every one-voter profile $\succ^{\{i\}}\in\mathcal P^{\{i\}}(B)$,
$aP_{F(\succ^{\{i\}})}b \Longleftrightarrow a\succ_i b$ for every pair of distinct alternatives $a,b\in B$.

\medskip
\noindent
\textbf{Cancellation.}
For every $B\in\mathcal B$ and every $\succ\,\in\mathcal P(B)$, if $T(\succ)=0_B$, then $F(\succ)=(B)$.
\medskip

\textit{Neutrality} requires invariance with respect to the names of the alternatives.
\textit{Consistency} requires social comparisons supported by two disjoint electorates to be preserved when the electorates are combined.
\textit{Faithfulness} requires the social ranking in a one-voter problem to coincide with that voter's ranking.
\textit{Cancellation} requires complete social indifference whenever all pairwise majority margins are zero.

\section{Recursive Borda aggregation rules}
\label{sec:recursive-borda}

We now introduce a general class of recursive elimination rules.
We then specialize this class to procedures whose elimination decisions are based on Borda scores.

For every feasible set $B \in \mathcal{B}$, let $D_B : \mathcal{P}(B) \longrightarrow 2^{B}\backslash \{\emptyset\}$ be a correspondence.
For every profile $\succ \,\in \mathcal{P}(B)$, the set $D_B(\succ)$ is interpreted as the set of alternatives eliminated in the current round. The elimination set is allowed to coincide with the entire feasible set. The case $D_B(\succ) = B$ means that the procedure stops and places all remaining alternatives in the same final tier. Let $D = (D_B)_{B \in \mathcal{B}}$ denote the collection of elimination correspondences.

\begin{definition}\label{def:elim}
For every feasible set $B\in\mathcal{B}$, an \emph{elimination correspondence on $B$} is a
mapping $D_{B}:\mathcal{P}(B)\longrightarrow 2^{B}\setminus\{\emptyset\}$ satisfying
$\emptyset\neq D_{B}(\succ)\subseteq B$ for every $\succ\,\in\mathcal{P}(B)$. A
\emph{collection of elimination correspondences} is a family $D=(D_{B})_{B\in\mathcal{B}}$
of such mappings.
\end{definition}

The elimination set is allowed to
coincide with the entire feasible set: the case $D_{B}(\succ)=B$ means that the procedure
stops and places all remaining alternatives in the same final tier.

Fix a collection of elimination correspondences $D=(D_B)_{B\in\mathcal B}$.
Let us consider the following procedure.
\begin{itemize}
  \item For every profile $\succ \,\in \mathcal{P}(B)$, set $B^{0} = B$.

  \item At every round $t \ge 0$, define
  \[
    E^{t}_{D}(\succ) \;=\; D_{B^{t}}\!\left(\succ\!\!|_{B^{t}}\right).
  \]
  If $E^{t}_{D}(\succ) = B^{t}$, the procedure stops. Every alternative in $B^{t}$ is then assigned to the same final tier.

  \item If $E^{t}_{D}(\succ) \subsetneq B^{t}$, define
  \[
    B^{t+1} \;=\; B^{t} \setminus E^{t}_{D}(\succ)
  \]
  and proceed to the next round.
\end{itemize}
Because the feasible set is finite and at least one alternative is removed whenever the procedure continues, the procedure terminates after finitely many rounds.

For every alternative $a \in B$, define its \emph{terminal round} $\tau_{D}(a;\succ)$ by
\[
  \tau_{D}(a;\succ) = t
  \quad\Longleftrightarrow\quad
  a \in E^{t}_{D}(\succ).
\]
Since the procedure terminates and the elimination sets partition $B$, the terminal round is well defined for every $a \in B$.


The \emph{recursive elimination rule generated by $D$}, denoted by
$F^D$, is defined by
\[
aR_{F^D(\succ)}b \Longleftrightarrow
\tau_D(a;\succ)\geq\tau_D(b;\succ).
\]

Thus, alternatives assigned to the same terminal round are socially
indifferent, whereas alternatives surviving longer are ranked higher.

By construction,
\[
D_B(\succ)=W_{F^D}(\succ),
\]
where $W_F(\succ)=\operatorname{bottom}(F(\succ))$ for every profile $\succ\,\in\mathcal P(B)$ as a notational convenience.
Moreover, if
$D_B(\succ)\subsetneq B$, then
\[
F^D(\succ)|_{B\setminus D_B(\succ)}
=
F^D\bigl(\succ\!\!|_{B\setminus D_B(\succ)}\bigr).
\]
Thus, after the first elimination set has been removed, the continuation ranking is obtained by applying the same procedure to the restricted profile.

\subsection{Recursive Borda rules}
\label{subsec:recursive-borda-definition}

We next restrict attention to elimination procedures whose current
decision is determined by the centered Borda-score vector.
A recursive elimination rule $F^D$ is called a \emph{recursive Borda rule} if its generating collection of
elimination correspondences $D=(D_B)_{B \in \mathcal{B}}$ satisfies Neutrality, Borda Reducibility, Homogeneity, and Minimum-Score Elimination, which are defined as follows.
\begin{itemize}
\item[] \textbf{Neutrality.}
For every feasible set $B\in\mathcal B$, every profile $\succ\,\in\mathcal P(B)$, and every
$\sigma\in\Pi(A)$,
\[
D_{\sigma(B)}(\sigma(\succ))=\sigma\bigl(D_B(\succ)\bigr).
\]
\item[] \textbf{Borda Reducibility.}
For every feasible set $B\in\mathcal B$ and every pair of profiles
$\succ,\succ'\,\in\mathcal P(B)$,
\[
\beta^B(\succ)=\beta^B(\succ')
\quad\Longrightarrow\quad
D_B(\succ)=D_B(\succ').
\]
\item[] \textbf{Homogeneity.}
For every $B\in \mathcal B$, every $\succ\,\in\mathcal P(B)$, and every
$k\in\mathbb N$, $D_B(k*\!\succ)=D_B(\succ)$.
\item[] \textbf{Minimum-Score Elimination.}
For every feasible set $B\in\mathcal B$ and every profile
$\succ\,\in\mathcal P(B)$ such that $\beta^B(\succ)\neq\bm 0_B$,
\[
\arg\min_{a\in B}\beta_a^B(\succ)\subseteq D_B(\succ)\subsetneq B .
\]
\end{itemize}

Neutrality requires the elimination decision to be invariant with respect to the names of the
alternatives.
Borda Reducibility requires the current elimination decision to depend
only on the centered Borda-score vector.
Homogeneity requires that the elimination set be the same under the replication of electorates.
Minimum-Score Elimination requires the procedure to eliminate every alternative with the lowest
Borda score, and not to stop, whenever the score vector is nonzero. This is a minimal requirement
for an active elimination rule. Indeed, if at least one alternative must be eliminated, the
lowest-scoring alternatives are the most natural candidates for elimination.

Note that these conditions already imply that the procedure stops exactly when all remaining
alternatives have the same Borda score.

\begin{lemma}\label{lem:zst}
Let $F^{D}$ be a recursive Borda rule. Then, for every feasible set $B\in\mathcal B$ and every
profile $\succ\,\in\mathcal P(B)$,
\[
\beta^B(\succ)=\bm 0_B \quad\Longleftrightarrow\quad D_B(\succ)=B .
\]
\end{lemma}

\begin{proof}
The implication from right to left is Minimum-Score Elimination. For the converse, suppose that
$\beta^B(\succ)=\bm 0_B$ and let $\sigma\in\Pi(A)$ satisfy $\sigma(B)=B$. Then
$\beta^B(\sigma(\succ))=\sigma\bigl(\beta^B(\succ)\bigr)=\bm 0_B$, so Borda Reducibility gives
$D_B(\sigma(\succ))=D_B(\succ)$, while Neutrality gives $D_B(\sigma(\succ))=\sigma\bigl(D_B(\succ)\bigr)$.
Hence $\sigma\bigl(D_B(\succ)\bigr)=D_B(\succ)$ for every permutation $\sigma$ of $B$. The only
subsets of $B$ with this property are $\emptyset$ and $B$, and $D_B(\succ)\neq\emptyset$ by
definition. Therefore $D_B(\succ)=B$.
\end{proof}

We now define three canonical elimination correspondences.
For every profile $\succ\,\in\mathcal P(B)$, let
\[
  D^{\mathrm{Baldwin}}_{B}(\succ) =
    \displaystyle\operatorname*{arg\,min}_{a \in B} \beta^{B}_{a}(\succ),
\]
\[
  D^{\mathrm{sN}}_{B}(\succ) =
  \begin{cases}
    \bigl\{\, a \in B : \beta^{B}_{a}(\succ) < 0 \,\bigr\}
      & \text{if } \beta^{B}(\succ) \neq \bm 0_{B}, \\[1ex]
    B & \text{if } \beta^{B}(\succ) = \bm 0_{B},
  \end{cases}
\]
and
\[
  D^{\mathrm{Nanson}}_{B}(\succ) =
    \bigl\{\, a \in B : \beta^{B}_{a}(\succ) \le 0 \,\bigr\}.
\]

The recursive elimination rules generated by these correspondences are called the Baldwin rule, the strict Nanson rule, and the Nanson rule, respectively, and we denote them by
$F^{\mathrm{Baldwin}}$, $F^{\mathrm{sN}}$, and $F^{\mathrm{Nanson}}$.

We can see that each recursive Borda rule is a margin matrix rule. Indeed,
let $\succ, \succ' \,\in \mathcal{P}(B)$ satisfy $T(\succ) = T(\succ')$ and
argue by induction on $|B|$. By definition, their centered Borda-score vectors coincide. Borda Reducibility therefore implies that the two profiles generate the same first elimination set $E$.
If $E = B$, both procedures stop and produce complete indifference.
Otherwise, the same alternatives $C = B \setminus E$ remain after the first
round, and $T(\succ\!\!|_{C}) = T(\succ)|_{C} = T(\succ')|_{C} = T(\succ'\!\!|_{C})$
because restriction of a profile corresponds to restriction of its
majority-margin matrix. The induction hypothesis applied to $C$ therefore
gives $F^{D}(\succ\!\!|_{C}) = F^{D}(\succ'\!\!|_{C})$, and hence
$F^{D}(\succ) = F^{D}(\succ')$.


We illustrate how the elimination procedure works.
In particular, Borda Reducibility concerns the elimination decision in the current
round.
It does not imply that the complete recursive ranking is determined by
the initial Borda-score vector.
The reason is that restricting the feasible set changes the scores of the surviving alternatives according to their pairwise margins against the removed alternatives.
More precisely, let $\succ\,\in\mathcal P(B)$, let $\emptyset\neq C\subseteq B$, and let
$a\in C$.
Then
\[
\beta_a^C(\succ\!\!|_C)
=
\beta_a^B(\succ)
-
\sum_{z\in B\setminus C}
[T(\succ)]_{az}.
\]
That is, the centered Borda score on $B$ is the sum of the margins of $a$ against all alternatives in $B\setminus\{a\}$.
The score on $C$ omits exactly the margins against the alternatives in
$B\setminus C$.

\begin{example}[Identical initial scores and different recursive rankings]
\label{ex:initial-borda-scores}

Let $B=\{x,y,z\}$ and let $N=\{1,2,3\}$.
Consider the following two profiles.

\begin{table}[H]
\centering
\begin{tabular}{c|ccc|ccc}
\toprule
&
\multicolumn{3}{c|}{Profile $\succ^1$}
&
\multicolumn{3}{c}{Profile $\succ^2$}
\\
Rank
&
Voter 1
&
Voter 2
&
Voter 3
&
Voter 1
&
Voter 2
&
Voter 3
\\
\midrule
1
&
$x$
&
$x$
&
$y$
&
$x$
&
$y$
&
$y$
\\
2
&
$y$
&
$y$
&
$z$
&
$z$
&
$x$
&
$x$
\\
3
&
$z$
&
$z$
&
$x$
&
$y$
&
$z$
&
$z$
\\
\bottomrule
\end{tabular}
\caption{Profiles with identical initial centered Borda scores}
\label{tab:identical-initial-borda}
\end{table}

In the coordinate order $(x,y,z)$, both profiles have the same initial
centered Borda-score vector:
\[
\beta^B(\succ^1)
=
\beta^B(\succ^2)
=
(2,2,-4).
\]

Consequently, the Baldwin, strict Nanson, and Nanson correspondences
all eliminate $z$ in the first round at both profiles.

After $z$ has been removed, however, the reduced score vectors are $\beta^{\{x,y\}}(\succ^1\!\!|_{\{x,y\}})=(1,-1)$ and $\beta^{\{x,y\}}(\succ^2\!\!|_{\{x,y\}})=(-1,1)$.
Hence, for every
$F\in
\{F^{\mathrm{Baldwin}},F^{\mathrm{sN}},F^{\mathrm{Nanson}}\}$, we have 
\[
xP_{F(\succ^1)}yP_{F(\succ^1)}z \quad \text{and} \quad yP_{F(\succ^2)}xP_{F(\succ^2)}z.
\]

To see the source of the difference, note that at $\succ^1$, $[T(\succ^1)]_{xz}=1$ and $[T(\succ^1)]_{yz}=3$, whereas at $\succ^2$, $[T(\succ^2)]_{xz}=3$ and $[T(\succ^2)]_{yz}=1$.
Removing $z$ therefore lowers $y$'s score more than $x$'s score at $\succ^1$, while it lowers $x$'s score more than $y$'s score at $\succ^2$.
\end{example}

Example~\ref{ex:initial-borda-scores} clarifies the scope of Borda
Reducibility.
The initial Borda-score vector determines the current elimination set,
but it does not determine the sequence of Borda-score vectors generated
after successive restrictions.
The complete recursive ranking may therefore depend on pairwise-margin information that is not contained in the initial Borda-score vector.

\subsection{A foundation for recursive Borda rules}
\label{sec:recursive-borda-foundation}

We develop an axiomatic foundation for recursive Borda rules.
In particular, we consider weakenings of the axioms developed by \cite{Young1974} to see what properties allow us to use non-trivial elimination procedures with natural regularity properties.

The following axiom is the core axiom for the recursive elimination procedure, originally introduced by \cite{FreemanBrillConitzer2014} for their characterization of the Baldwin rule.

\medskip
\noindent
\textbf{Bottom Independence.}
For every feasible set $B\in\mathcal B$ and every profile $\succ\,\in\mathcal P(B)$ such that $W_F(\succ)\neq B$, $F(\succ)|_{B\setminus W_F(\succ)}=F\bigl(\succ\!\!|_{B\setminus W_F(\succ)}\bigr)$.
\medskip

\textit{Bottom Independence} requires the ranking among the surviving alternatives to coincide with the ranking obtained by applying the rule directly to the restricted profile.
It therefore turns the local bottom-tier decision into a recursive procedure.
The following lemma shows that this condition indeed captures the idea of recursive elimination rules.

\begin{lemma}
\label{lem:recursive-elimination}
A voting rule $F$ is a recursive elimination rule if and only if it satisfies \textit{Bottom Independence}.
\end{lemma}

\begin{proof}
Suppose first that $F$ is a recursive elimination rule generated by a collection of elimination correspondences $D$.
By construction, $D_B(\succ)=W_F(\succ)$.
If $D_B(\succ)\subsetneq B$ and $C=B\setminus D_B(\succ)$, the recursive construction gives $F(\succ)|_C=F(\succ\!\!|_C)$.
Since $C=B\setminus W_F(\succ)$, $F$ satisfies \textit{Bottom Independence}.

Conversely, suppose that $F$ satisfies \textit{Bottom Independence}.
Define the elimination correspondence
\[
D^F_B(\succ):=W_F(\succ).
\]
We show by induction on $|B|$ that $F=F^{D^F}$.
If $W_F(\succ)=B$, then all alternatives are socially indifferent, and hence $F(\succ)=(B)=F^{D^F}(\succ)$.
Suppose instead that $W_F(\succ)\subsetneq B$ and let $C=B\setminus W_F(\succ) \neq \emptyset$.
\textit{Bottom Independence} gives $F(\succ)|_C=F(\succ\!\!|_C)$.
By the induction hypothesis,
\[
F(\succ\!\!|_C)=F^{D^F}(\succ\!\!|_C).
\]
Moreover, both $F(\succ)$ and $F^{D^F}(\succ)$ place the alternatives $D^F_B(\succ)=W_F(\succ)$ together in the bottom indifference class and place every alternative in $C$ strictly above them.  Therefore, $F(\succ)=F^{D^F}(\succ)$.
Thus, $F$ is a recursive elimination rule.
\end{proof}

Lemma~\ref{lem:recursive-elimination} separates the recursive structure of a ranking rule from the criterion used at each stage.
\textit{Bottom Independence} is equivalent to the existence of a recursive
elimination representation, and the elimination correspondence is
uniquely given by the bottom correspondence $W_F(\succ)$.
Accordingly, once \textit{Bottom Independence} is imposed, the remaining task is to characterize the one-step correspondence $W_F$ that determines
which alternatives are eliminated at the current stage.

We now consider how this recursive structure interacts with electorate aggregation.  
A natural benchmark is Young's \textit{Consistency}, which requires the complete social ranking to respond
coherently when disjoint electorates are combined.  
The following result shows, however, that this  requirement is incompatible with recursive Borda aggregation.

\begin{proposition}
\label{prop:recursive-borda-inconsistency}
No recursive Borda rule satisfies \textit{Consistency}.
\end{proposition}

\begin{proof}
Let $B=\{a,b,c\}$ and consider two profiles on disjoint electorates:
\[
\succ^0
=
4[b\succ a\succ c]
+
4[c\succ b\succ a]
+
4[a\succ c\succ b],
\]
and
\[
\succ^1
=
3[a\succ b\succ c]
+
[c\succ a\succ b]
+
2[b\succ c\succ a],
\]
where each coefficient denotes the number of voters having the corresponding preference ordering.
In the coordinates $(ab,ac,bc)$, their pairwise majority-margin
vectors are $T[\succ^0]=(-4,4,-4)$ and $T[\succ^1]=(2,0,4)$.
Consequently, their centered Borda-score vectors are $\beta_B(\succ^0)=(0,0,0)$ and $\beta_B(\succ^1)=(2,2,-4)$.

Since $\succ^0$ has the zero score vector, zero-score termination of the rule gives
\begin{equation}\label{eq:succ0}
F(\succ^0)=(B).
\end{equation}
Moreover, $T[\succ^0\cup \succ^1]=(-2,4,0)$, and hence $\beta_B(\succ^0\cup \succ^1)=(2,2,-4)=\beta_B(\succ^1)$.
Neutrality, Minimum-Score Elimination, and Lemma~\ref{lem:zst} therefore imply that the current bottom tier is $W_F(\succ^1)=W_F(\succ^0\cup \succ^1)=\{c\}$.

Consider now the recursive continuation after $c$ has been removed.
On $\{a,b\}$, the majority margin at $\succ^1$ is $T(\succ^1)_{ab}=2$, whereas the majority margin at the combined profile is $T[\succ^0\cup \succ^1]_{ab}=-2$.
Thus, $F(\succ^1\!\!|_{\{a,b\}})=a\succ b$, while $F((\succ^0\cup \succ^1)|_{\{a,b\}})=b\succ a$.
Since $F$ is recursively generated by elimination of its bottom tier, or equivalently satisfies \textit{Bottom Independence},
\begin{equation}\label{eq:succ1}
F(\succ^1)=a\succ b\succ c
\end{equation}
and
\begin{equation}\label{eq:succ0+1}
F(\succ^0\cup \succ^1)=b\succ a\succ c.
\end{equation}

On the other hand, (\ref{eq:succ0}) and Consistency imply $F(\succ^0\cup \succ^1)=F(\succ^1)$, contradicting (\ref{eq:succ1}) and (\ref{eq:succ0+1}).
Therefore, no recursive Borda rule satisfies \textit{Consistency}.
\end{proof}

The contradiction does not arise at the first elimination stage.
Indeed,
\[
W_F(\succ^1)
=
W_F(\succ^0\cup \succ^1)
=
\{c\}.
\]
Thus, adding $\succ^0$ does not change the alternative currently eliminated.
The contradiction arises only after $c$ has been removed.  
Although $\succ^0$ has zero centered Borda scores on the original set $\{a,b,c\}$, it contains pairwise information concerning $a$ and $b$.  
This information becomes relevant in the restricted problem on $\{a,b\}$ and reverses their ranking.
\textit{Consistency} requires $\succ^0$, which produces complete social indifference on the original problem, to be irrelevant to the entire social ranking.  
Recursive elimination instead requires the rule to use the restriction of $\succ^0$ once the current bottom alternative has been removed. 
These two requirements cannot be reconciled.

Proposition~\ref{prop:recursive-borda-inconsistency} does not claim
that \textit{Consistency} and \textit{Bottom Independence} are logically incompatible for arbitrary ranking rules.  The impossibility is specific to their
combination with Borda-based recursive elimination: a profile that is irrelevant according to the Borda scores on the original feasible set may become relevant after an alternative has been removed.

The proof also indicates how \textit{Consistency} should be weakened.
The current bottom tier is unaffected by the addition of $\succ^0$:
\[
W_F(\succ^1)
=
W_F(\succ^0\cup \succ^1)
=
\{c\}.
\]
What cannot be preserved is the complete ranking of the alternatives that survive the current elimination stage.
Accordingly, we retain consistency only for the current bottom
correspondence. 
The resulting axiom does not require the entire social ranking to be invariant under electorate aggregation. 
It requires only the current elimination decision to respond consistently.
Therefore, we consider the following weaker form of \textit{Consistency}.

\medskip
\noindent
\textbf{Bottom Consistency.}
For every feasible set $B\in\mathcal B$, every pair of disjoint electorates $N,N'\in\mathcal U$, and every pair of profiles $\succ^N\in\mathcal P^N(B)$ and
$\succ^{N'}\in\mathcal P^{N'}(B)$, the following conditions hold.

\begin{enumerate}[(i)]
\item
If $W_F(\succ^N) \cap W_F(\succ^{N'}) \neq \emptyset$, then $W_F(\succ^N)\cap W_F(\succ^{N'})\subseteq W_F(\succ^N\cup\succ^{N'})$.

\item
If $W_F(\succ^N)=B$, then $W_F(\succ^N\cup\succ^{N'})=W_F(\succ^{N'})$.
The symmetric requirement applies when $W_F(\succ^{N'})=B$.

\item
If $W_F(\succ^N)=W_F(\succ^{N'})$, then $W_F(\succ^N\cup\succ^{N'})=W_F(\succ^N)$.
\end{enumerate}

The first condition, (i), preserves alternatives that both electorates place in
the current bottom tier.
The second condition, (ii), requires an electorate that is completely
indifferent to have no effect on the bottom tier of another electorate.
The third condition, (iii), preserves a common bottom tier exactly.
The qualifications in~(i) are stated only for transparency, since the
requirement is vacuous when the intersection is empty and already follows
from~(ii) when the intersection equals $B$.

The asymmetry between the inclusion in~(i) and the equalities in~(ii) and~(iii)
is deliberate, and reflects the passage from pairwise comparisons to a
set-valued elimination decision.
Young's \textit{Consistency} is a requirement about agreement, in that a social
comparison endorsed by two disjoint electorates must survive their union.
It is silent about comparisons on which the two electorates disagree.
For the bottom correspondence, the two electorates agree about an alternative
precisely when that alternative belongs to
$W_F(\succ^N)\cap W_F(\succ^{N'})$, and the inclusion in~(i) states that such
agreement is preserved under aggregation.
Imposing equality would add the upper bound
\[
W_F(\succ^N\cup\succ^{N'})\ \subseteq\ W_F(\succ^N)\cap W_F(\succ^{N'}),
\]
which has no counterpart in \textit{Consistency}, since it legislates on
disagreement.
Whenever one electorate would eliminate an alternative and the other would
retain it, the upper bound forces the combined electorate to retain it.
Preserving agreement and resolving disagreement in a designated direction are
separate normative demands, and only the former belongs to Young's aggregation
principle.

The distinction also affects what the axiom governs.
Under equality, the bottom tier can never grow when electorates are combined.
A consistency requirement would then also settle how many alternatives a single
round may remove, since the bottom tier of a large electorate could never be
coarser than that of any of its parts.
An axiom intended to govern the response of the elimination decision to the
composition of the electorate would thereby also fix the coarseness of the
elimination step itself.\footnote{The equality version of~(i) is
\textit{Strong Bottom Consistency}, introduced in Section~\ref{sec:canonical-axiomatizations}. Theorem~\ref{thm:unified-characterization}~(i)
shows that, together with the remaining axioms of Theorem~1, it determines a
single rule.}
How coarsely a rule eliminates at each stage is a question about the
elimination criterion rather than about aggregation, and it is better addressed
by an explicit axiom whose content can be assessed on its own terms.
Condition~(i) therefore retains the lower bound alone.

No such issue arises for~(ii) and~(iii).
In both cases the intersection coincides with one of the two bottom tiers, with
$W_F(\succ^{N'})$ under~(ii) and with the common tier under~(iii), so that the
two electorates do not disagree about which alternatives to eliminate.
There being no disagreement to resolve, the upper bound carries no additional
content, and equality is the natural formulation.


Young's \textit{Faithfulness} requires the complete social ranking in a one-voter problem to reproduce that voter's ranking.
The following weak version, \textit{Weak Faithfulness}, requires the voter's top-ranked alternative to survive the first elimination round.

\medskip
\noindent
\textbf{Weak Faithfulness.}
For every feasible set $B\in\mathcal B$ with $|B|\geq2$ and every
one-voter profile $\succ^{\{i\}}\in\mathcal P^{\{i\}}(B)$, if $x\succ_i y$ for every $y\in B\setminus\{x\}$, then $x\notin W_F(\succ^{\{i\}})$.
\medskip

The following result shows that \textit{Neutrality}, \textit{Cancellation}, \textit{Bottom Consistency} and \textit{Weak Faithfulness} jointly determine the local Borda structure of the rule.

\begin{proposition}
\label{prop:local-borda-structure}
Suppose that a voting rule $F$ satisfies \textit{Neutrality}, \textit{Bottom Consistency}, \textit{Weak Faithfulness}, and \textit{Cancellation}.
Then, for every feasible set $B\in\mathcal B$, there exists a well-defined correspondence $\Delta_B:\mathcal S(B)\longrightarrow2^B\setminus\{\emptyset\}$
with the following properties.

\begin{enumerate}[(i)]

\item
For every profile $\succ\,\in\mathcal P(B)$, $W_F(\succ)=\Delta_B\bigl(\beta^B(\succ)\bigr)$.

\item
For every permutation $\sigma\in\Pi(A)$ and every
$s\in\mathcal S(B)$, $\Delta_{\sigma(B)}(\sigma(s))=\sigma\bigl(\Delta_B(s)\bigr)$.

\item
For every $s\in\mathcal S(B)$ and every positive rational number $\lambda$, $\Delta_B(\lambda s)=\Delta_B(s)$.

\item
Complete indifference occurs exactly at the zero score vector: $\Delta_B(s)=B\Longleftrightarrow
s=\bm 0_B$.

\item 
The bottom tier is a lower interval of the centered Borda-score ordering: there exists a non-positive real number $c_B(s) \in [\min_{a\in B}s_a, 0]$ such that $\Delta_B(s)=\{a \in B : s_a \le c_B(s)\}$.

\item 
For every $s,t\in S(B)$, if $\Delta_B(s)\cap\Delta_B(t)\neq\emptyset$, then
$\Delta_B(s)\cap\Delta_B(t) \subseteq \Delta_B(s+t)$, with equality whenever $\Delta_B(s)=\Delta_B(t)$.
\end{enumerate}
\end{proposition}

The proof has five steps.
First, \textit{Bottom Consistency} and \textit{Cancellation} imply that the bottom
tier depends only on the majority-margin matrix and is invariant under replication.
Second, \textit{Neutrality} implies that cyclic components of a majority-margin
matrix generate complete social indifference.
The bottom tier therefore depends only on the cocyclic component and hence only on the centered Borda-score vector.
Third, \textit{Neutrality}, \textit{Bottom Consistency}, \textit{Weak Faithfulness}, and \textit{Cancellation} determine the bottom tiers of the canonical positive and negative star vectors.
Fourth, symmetrizing an arbitrary score vector over the permutations fixing an alternative with a positive score yields a positive multiple of the corresponding positive star vector, whose bottom tier excludes that alternative.
Finally, these star vectors and \textit{Bottom Consistency} imply that the bottom tier is downward closed in the Borda-score ordering.
Hence no alternative with a positive score can belong to the bottom tier, and the nondegeneracy in (iv) follows.
Condition (vi) is the score-space counterpart of  \textit{Bottom Consistency} under the score-space representation established in (i).
Indeed, by part (i), the bottom tier depends only on the associated score vectors.
Part (i) of \textit{Bottom Consistency} yields the inclusion in (vi), while part (iii) yields equality whenever $\Delta_B(s)=\Delta_B(t)$.
Conversely, (vi) immediately implies \textit{Bottom Consistency} after replacing score vectors by realizing profiles.
The complete proof is provided in
Appendix~\ref{app:local-borda-structure}.

Proposition~\ref{prop:local-borda-structure}  completely characterizes the current bottom-tier correspondence in score space.
\textit{Bottom Independence} then determines the continuation ranking once that bottom tier has been removed. Combining the two yields the following characterization of regular recursive Borda rules.


\begin{definition}\label{def:regular}
A recursive Borda rule $F^{D}$ is \emph{regular} if its generating collection of elimination correspondences $D=(D_{B})_{B\in\mathcal B}$ satisfies the conditions of Proposition~\ref{prop:local-borda-structure}.
\end{definition}

\begin{theorem}\label{thm:recursive-borda}
A voting rule $F:\mathcal P\to\mathcal R$ satisfies \textit{Neutrality}, \textit{Bottom Consistency}, \textit{Weak
Faithfulness}, \textit{Cancellation}, and \textit{Bottom Independence} if and only if $F$ is a regular recursive Borda
rule.
\end{theorem}

\begin{proof}
Necessity follows directly from the definition of a regular recursive Borda rule together with Bottom Independence.
We show sufficiency.
Suppose that $F$ satisfies the five axioms.
Proposition~\ref{prop:local-borda-structure} gives the Borda dependence, Neutrality, positive Homogeneity,
nondegeneracy, the lower-interval structure, and the zero upper bound of the
current bottom tier, $W_F$.
Therefore, $D=(D_B)_{B\in\mathcal B}$ satisfies Neutrality, Borda Reducibility, Homogeneity, and Minimum-Score Elimination.
\textit{Bottom Independence} directly implies that $F$ is a recursive elimination rule with respect to $D=(D_B)_{B \in \mathcal{B}}$.
Since $D_B(\succ)=W_{F^D}(\succ)$, for every profile $\succ\,\in\mathcal P(B)$, the above arguments show that $D=(D_B)_{B\in\mathcal B}$ satisfies the conditions of Proposition~\ref{prop:local-borda-structure}, so that $F$ is a regular recursive Borda rule.
\end{proof}

\section{The Condorcet-consistent aggregation}
\label{sec:condorcet-safe-interval}

We now examine the restriction imposed by Condorcet Consistency on the elimination set.
An alternative $x\in B$ is a \emph{Condorcet winner} at $\succ\,\in\mathcal P(B)$ if $[T(\succ)]_{xy}>0$ for every $y\in B\setminus\{x\}$.
A voting rule $F$ satisfies \emph{Condorcet Consistency} if $\operatorname{top}(F(\succ))=\{x\}$ whenever $x$ is a Condorcet winner at $\succ$.
The key observation is that any alternative with a positive centered Borda score can be made a Condorcet winner without changing the
Borda-score vector.

\begin{lemma}
\label{lem:condorcet-completion}

Let $B$ be a finite set of alternatives and let $s=(s_a)_{a\in B}\in\mathbb Q^B$ satisfy
$\sum_{a\in B}s_a=0$.
For every $x\in B$ such that $s_x>0$, there exists a rational skew-symmetric matrix $T\in\mathbb Q^{B\times B}$ whose row-sum vector is $s$ and such that $T_{xy}>0$ for every $y\in B\setminus\{x\}$.
\end{lemma}

\begin{proof}
Choose positive rational numbers
$(\varepsilon_y)_{y\in B\setminus\{x\}}$
satisfying $\sum_{y\neq x}\varepsilon_y=s_x$.
Set $T_{xy}=\varepsilon_y$ and $T_{yx}=-\varepsilon_y$ for every $y\neq x$.

Let $C=B\setminus\{x\}$ and define
$r_y=s_y+\varepsilon_y$ for every $y\in C$.
Then $\sum_{y\in C}r_y=0$.
Let $q=|C|$ and define
\[
T_{yz}
=
\frac{r_y-r_z}{q}
\]
for every $y,z\in C$.
The resulting matrix is skew-symmetric.
For every $y\in C$, its row sum within $C$ is $r_y$.
Its full row sum is therefore $r_y-\varepsilon_y=s_y$ and the row sum of $x$ is
$\sum_{y\neq x}\varepsilon_y=s_x$.
Finally, $T_{xy}>0$ for every $y\neq x$.
\end{proof}

We show that \textit{Condorcet Consistency} requires the centered Borda score of every eliminated alternative to be weakly below the average, which implies the following statements.

\begin{theorem}\label{thm:condorcet-safe}
A recursive Borda rule satisfies \textit{Condorcet Consistency} if and only if
\[
D_{B}(\succ)\subseteq\{a\in B:\beta^{B}_{a}(\succ)\le 0\}
\]
for every feasible set $B\in\mathcal B$ and every profile $\succ\in\mathcal P(B)$.
In particular, every regular recursive Borda rule satisfies \textit{Condorcet Consistency}. 
\end{theorem}

\begin{proof}
Suppose first that a recursive Borda rule $F^D$ satisfies \textit{Condorcet Consistency}.
Assume to the contrary that there exist $B\in\mathcal B$, $\succ\,\in\mathcal P(B)$, and $x\in D_B(\succ)$ such that $\beta_x^B(\succ)>0$.
Let $s=\beta^B(\succ)$.
By Lemma~\ref{lem:condorcet-completion}, there exists a rational skew-symmetric matrix $T'$ with row-sum vector $s$ at which $x$ is a Condorcet winner.

Choose a positive integer $k$ such that every off-diagonal entry of $kT'$ is an even integer.
By the majority-margin realizability theorem of \citet{Debord1987}, there exists
a profile $\succ'\,\in\mathcal P(B)$ that satisfies $T(\succ')=kT'$.
Since $k$ is positive, the alternative $x$ is a Condorcet winner at $\succ'$.
Moreover, the centered Borda-score vector is the row-sum vector of the majority-margin matrix, so that
$\beta^B(\succ')=k\beta^B(\succ)=\beta^B(k*\!\!\succ)$.
Borda Reducibility implies $D_B(\succ')=D_B(k*\!\!\succ)$ and Homogeneity implies
$D_B(k*\!\!\succ)=D_B(\succ)$.
Hence, $x\in D_{B}(\succ')$. Moreover, since $\beta^{B}_{x}(\succ')=k\beta^{B}_{x}(\succ)>0$,
we have $\beta^{B}(\succ')\neq \bm{0}_{B}$, and Minimum-Score Elimination gives
$D_{B}(\succ')\subsetneq B$.
Thus, the procedure does not stop at the first round, and the recursive procedure eliminates a Condorcet winner in the first round, contradicting \textit{Condorcet Consistency}.
It follows that every eliminated alternative has a nonpositive centered Borda score.

Conversely, suppose that the stated inclusion holds.
Let $x$ be a Condorcet winner in $\succ \, \in\mathcal P(B)$.
For every $C\subseteq B$ containing $x$ with
$|C|\geq2$, the alternative $x$ remains a Condorcet winner at $\succ\!\!|_C$.
Therefore, $\beta_x^C(\succ\!\!|_C)>0$.
The assumed inclusion implies that $x\notin D_C(\succ\!\!|_C)$.
When $x$ is the only remaining alternative, its centered Borda score is zero, and the procedure stops.
Hence, $x$ is uniquely top-ranked, so that \textit{Condorcet Consistency} holds.
\end{proof}


For two collections of elimination correspondences $D$ and $D'$, write
$D\preceq D'$ if $D_B(\succ)\subseteq D'_B(\succ)$ for every $B\in\mathcal B$ and every
$\succ\,\in\mathcal P(B)$.
Minimum-Score Elimination and Theorem \ref{thm:condorcet-safe} imply that the Baldwin and the Nanson rules are the unique minimal and maximal rules in the class of Condorcet-consistent recursive Borda rules.

\begin{corollary}
\label{cor:baldwin-nanson-interval}
For any Condorcet-consistent recursive Borda rule $F^D$ generated by $D$, it holds that 
\[
D^{\mathrm{Baldwin}}
\preceq
D
\preceq
D^{\mathrm{Nanson}}.
\]
\end{corollary}

The strict Nanson rule is a natural interior member of the interval.
It eliminates all alternatives whose Borda scores are strictly below the average while retaining alternatives whose scores equal the
average.
The interval generally contains other recursive Borda rules as well.

The class of recursive Borda rules considered in Corollary \ref{cor:baldwin-nanson-interval} is still quite large even if elimination correspondences have interval structures.
For example, the recursive Borda rule generated by the following elimination correspondences $D^\ast$ belongs to this class: 
\[
D_B^\ast(\succ)
=
\begin{cases}
D_B^{\mathrm{Baldwin}}(\succ)~\text{if}~ |B| \text{ is odd},\\
D_B^{\mathrm{sN}}(\succ)~ \text{if } |B| \text{ is even}.
\end{cases}
\]

To identify canonical members of the interval characterized in Corollary \ref{cor:baldwin-nanson-interval}, we impose a uniformity condition linking current bottom-tier decisions across profiles and feasible sets.
For every real number $r$, define
\[
\operatorname{sgn}(r)
=\begin{cases}
-1~\text{if}~r<0,\\
0~\text{if}~r=0,\\
1,~\text{if}~r>0.
\end{cases}
\]
For every profile $\succ\,\in\mathcal P(B)$, let
\[
M_B(\succ)=
\arg\min_{a\in B}\beta_a^B(\succ).
\]

\begin{itemize}
\item[] \textbf{Nonminimum Sign-Threshold Uniformity.} For every pair of feasible sets $B,B'\in\mathcal B$, every pair of profiles $\succ\,\in\mathcal P(B)$ and $\succ'\,\in\mathcal P(B')$, and every pair of alternatives $a\in B\setminus M_B(\succ)$ and $a'\in B'\setminus M_{B'}(\succ')$, if $\operatorname{sgn}\bigl(\beta_a^B(\succ)\bigr)\leq\operatorname{sgn}\bigl(\beta_{a'}^{B'}(\succ')\bigr)$ and $a'\in W_F(\succ')$, then $a\in W_F(\succ)$.
\end{itemize}

The condition applies both within and across profiles.
It requires the treatment of nonminimum alternatives to be governed by a common cutoff in the sign ordering.
In particular, alternatives with the same sign are treated identically.
Moreover, if an alternative with a given sign belongs to the bottom tier, then every nonminimum alternative with a lower sign also belongs to the bottom tier.
Minimum-score alternatives are excluded from the condition because their inclusion in the bottom tier is already implied by Minimum-Score Elimination.

\begin{corollary}
\label{cor:baldwin-nanson-extremes}
Suppose that $F$ is a regular recursive Borda rule and satisfies Nonminimum Sign-Threshold Uniformity.
Then $F$ is Baldwin, strict Nanson, or Nanson rule.
\end{corollary}

\begin{proof}
By Definition~\ref{def:regular} and Proposition~\ref{prop:local-borda-structure}~(v), every eliminated alternative has a nonpositive centered Borda score. 
Since every regular recursive Borda rule eliminates all minimum-score alternatives, Nonminimum Sign-Threshold Uniformity implies that the only possible elimination regions are three cases: $M_B(\succ)$, $\left\{ a\in B: \beta_a^B(\succ)<0\right\}$, and $\left\{a\in B: \beta_a^B(\succ)\leq 0 \right\}$.
These are precisely the elimination regions of the Baldwin, strict Nanson, and Nanson rules.  
\end{proof}

The inclusion relations among the classes identified in Theorem~\ref{thm:recursive-borda} and Corollaries \ref{cor:baldwin-nanson-interval} and \ref{cor:baldwin-nanson-extremes} are summarized below.

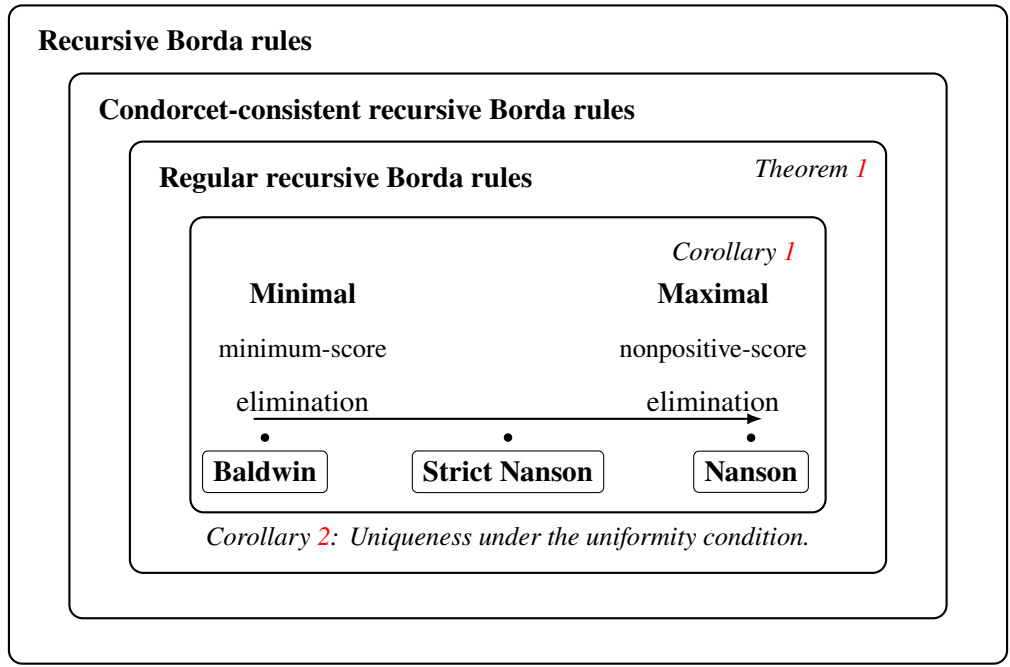
\begin{figure}[H]
\centering
\begin{tikzpicture}[
    font=\small,
    class/.style={
        draw,
        rounded corners=5pt,
        thick,
        align=center
    },
    rule/.style={
        draw,
        rounded corners=2pt,
        fill=white,
        align=center,
        inner sep=4pt
    },
    theorem/.style={
        font=\footnotesize\itshape,
        align=center
    }
]

\node[class,
      minimum width=13.2cm,
      minimum height=8.7cm,
      anchor=north west]
      (RB) at (0,0) {};

\node[class,
      minimum width=11.6cm,
      minimum height=7.2cm,
      anchor=north west]
      (RRB) at (0.8,-0.9) {};

\node[class,
      minimum width=10.0cm,
      minimum height=5.7cm,
      anchor=north west]
      (CCRB) at (1.6,-1.8) {};

\node[class,
      minimum width=8.4cm,
      minimum height=3.9cm,
      anchor=north west]
      (BNI) at (2.4,-2.8) {};

\node[anchor=north west, align=left]
      at ($(RB.north west)+(0.25,-0.20)$)
      {\textbf{Recursive Borda rules}};

\node[anchor=north west, align=left]
      at ($(RRB.north west)+(0.25,-0.20)$)
      {\textbf{Condorcet-consistent recursive Borda rules}};

\node[anchor=north west, align=left]
      at ($(CCRB.north west)+(0.25,-0.20)$)
      {\textbf{Regular recursive Borda rules}};

    \node[theorem, anchor=north east]
      at ($(RRB.north east)+(-0.9,-1.0)$)
      {Theorem~\ref{thm:recursive-borda}};

\node[theorem, anchor=north east]
      at ($(BNI.north east)+(-0.25,-0.20)$)
      {Corollary~\ref{cor:baldwin-nanson-interval}};

\draw[-{Latex[length=2mm]}, thick]
      ($(BNI.south west)+(0.85,1.25)$)
      --
      ($(BNI.south east)+(-0.85,1.25)$);

\node[align=center, anchor=south]
      at ($(BNI.south west)+(1.5,1.2)$)
      {\textbf{Minimal}\\
       \footnotesize minimum-score \\ elimination};

\node[align=center, anchor=south]
      at ($(BNI.south east)+(-1.5,1.2)$)
      {\textbf{Maximal}\\
       \footnotesize nonpositive-score \\ elimination};

\node[rule]
      (Baldwin)
      at ($(BNI.south west)+(1.00,0.55)$)
      {\textbf{Baldwin}};

\node[rule]
      (StrictNanson)
      at ($(BNI.south)+(0,0.55)$)
      {\textbf{Strict Nanson}};

\node[rule]
      (Nanson)
      at ($(BNI.south east)+(-1.00,0.55)$)
      {\textbf{Nanson}};

\node[theorem, below=0.3cm of StrictNanson]
      {Corollary~\ref{cor:baldwin-nanson-extremes}: Uniqueness under the uniformity condition.};

\fill ($(Baldwin.north)+(0,0.17)$) circle (1.6pt);
\fill ($(StrictNanson.north)+(0,0.17)$) circle (1.6pt);
\fill ($(Nanson.north)+(0,0.17)$) circle (1.6pt);

\end{tikzpicture}

\end{figure}

\section{Axiomatizations of the canonical rules}
\label{sec:canonical-axiomatizations}

Section \ref{sec:condorcet-safe-interval} identifies the Baldwin, strict Nanson, and Nanson rules as the canonical Condorcet-consistent rules in the class of regular recursive Borda rules.
We now characterize each of the three rules directly.

First, observe that the Baldwin rule satisfies a stronger version of \textit{Bottom Consistency}.

\medskip
\noindent
\textbf{Strong Bottom Consistency.}
For every feasible set $B\in\mathcal B$, every pair of disjoint electorates $N,N'\in\mathcal U$, and every pair of profiles $\succ^N\in\mathcal P^N(B)$ and $\succ^{N'}\in\mathcal P^{N'}(B)$, if $W_F(\succ^N)\cap W_F(\succ^{N'})\neq\emptyset$,
then $W_F(\succ^N\cup\succ^{N'})=W_F(\succ^N)\cap W_F(\succ^{N'})$.
\medskip

\textit{Strong Bottom Consistency} requires the bottom tier of the combined electorate to consist exactly of the alternatives placed in the bottom tier by both constituent electorates.

The strict Nanson rule satisfies a consistency condition concerning surviving alternatives.
For every profile $\succ\,\in\mathcal P(B)$, write
$S_F(\succ)=B\setminus W_F(\succ)$.

\medskip
\noindent
\textbf{Survival Consistency.}
For every feasible set $B\in\mathcal B$, every pair of disjoint
electorates $N,N'\in\mathcal U$, and every pair of profiles $\succ^N\in\mathcal P^N(B)$ and $\succ^{N'}\in\mathcal P^{N'}(B)$, if $F(\succ^N\cup\succ^{N'})\neq(B)$,
then $S_F(\succ^N)\cap S_F(\succ^{N'}) \subseteq S_F(\succ^N\cup\succ^{N'})$.
\medskip

\textit{Survival Consistency} preserves alternatives that survive in both
electorates.
It complements \textit{Bottom Consistency}, which preserves alternatives
eliminated in both electorates.
\medskip

We next distinguish the strict Nanson and Nanson rules through the
treatment of dummy alternatives.
An alternative $x\in B$ is a \emph{dummy} at a profile
$\succ\,\in\mathcal P(B)$ if $[T(\succ)]_{xy}=0$
for every $y\in B\setminus\{x\}$.

\medskip
\noindent
\textbf{Dummy Retention.}
For every feasible set $B\in\mathcal B$, every profile
$\succ\,\in\mathcal P(B)$, and every dummy alternative $x\in B$, if $x\in W_F(\succ)$, then 
$F(\succ)=(B)$.
\medskip

\textit{Dummy Retention} requires a dummy alternative to survive unless all alternatives are socially indifferent.

\medskip
\noindent
\textbf{Dummy Exclusion.}
For every feasible set $B\in\mathcal B$, every profile$\succ\,\in\mathcal P(B)$, and every dummy alternative $x\in B$, $x\in W_F(\succ)$.
\medskip

\textit{Dummy Exclusion} requires every dummy alternative to belong to the bottom tier.
Note that each dummy condition implies Cancellation.
Combined with the axioms of Theorem~\ref{thm:recursive-borda},
each of these axioms pins down the bottom tier as that of the Baldwin, strict Nanson, and Nanson rules, respectively.

\begin{proposition}
\label{prop:canonical-local-identification}
Let $F$ satisfy \textit{Neutrality}, \textit{Bottom Consistency}, \textit{Weak Faithfulness}, \textit{Cancellation}, and \textit{Bottom Independence}.
Let $\Delta_B$ be the score-space correspondence derived in Proposition~\ref{prop:local-borda-structure}.

\begin{enumerate}[(i)]

\item
If $F$ satisfies \textit{Strong Bottom Consistency}, then
\[
\Delta_B(s)
=
\arg\min_{a\in B}s_a
\]
for every $s\in\mathcal S(B)$.

\item
If $F$ satisfies \textit{Survival Consistency} and \textit{Dummy Retention}, then
\[
\Delta_B(s)
=\begin{cases}
\{a\in B:s_a<0\},~\text{if}~s\neq\bm 0_B,\\
B,~\text{if}~s=\bm 0_B.
\end{cases}
\]

\item
If $F$ satisfies \textit{Dummy Exclusion}, then
\[
\Delta_B(s)
=
\{a\in B:s_a\leq0\}
\]
for every $s\in\mathcal S(B)$.

\end{enumerate}

\end{proposition}

The proof uses the canonical star vectors established in the proof of Proposition~\ref{prop:local-borda-structure}.
For the Baldwin rule, every centered score vector can be decomposed into positive star vectors whose bottom tiers intersect exactly at the minimum-score set.
\textit{Strong Bottom Consistency} turns this intersection into the exact bottom tier.
For the strict Nanson rule, \textit{Dummy Retention} excludes zero-score alternatives from every nontrivial bottom tier.
\textit{Survival Consistency} then separates negative from nonnegative
coordinates.
For the Nanson rule, \textit{Dummy Exclusion} includes every zero-score alternative.
\textit{Bottom Consistency} propagates this inclusion from zero coordinates to all negative coordinates.
The complete proof is provided in Appendix~\ref{app:canonical-local-identification}.

Given the identification of Bottom tiers, \textit{Bottom Independence} implies the uniqueness of the social ranking, which is formally stated in the following lemma.

\begin{lemma}
\label{lem:bottom-tier-propagation}
Let $F$ and $G$ satisfy \textit{Bottom Independence}.
If $\operatorname{bottom}(F(\succ))=\operatorname{bottom}(G(\succ))$ for every feasible set and every profile, then $F=G$.
\end{lemma}

\begin{proof}
We proceed by induction on the number of alternatives.
The result is immediate for singleton feasible sets.

Fix $\succ\,\in\mathcal P(B)$ and let $W=
\operatorname{bottom}(F(\succ))=\operatorname{bottom}(G(\succ))$.
If $W=B$, both rules yield complete indifference.
Suppose that $W\neq B$, and let
$C=B\setminus W$.
\textit{Bottom Independence} gives $F(\succ)|_C=F(\succ\!\!|_C)$ and $G(\succ)|_C=G(\succ\!\!|_C)$.
Therefore, the induction hypothesis gives
$F(\succ)|_{C}=F(\succ\!\!|_C)=G(\succ\!\!|_C)=G(\succ)|_{C}$.
\end{proof}

By Theorem \ref{thm:recursive-borda}, Proposition \ref{prop:canonical-local-identification} and Lemma \ref{lem:bottom-tier-propagation}, we can obtain the unified axiomatization result of the three rules.

\begin{theorem}
\label{thm:unified-characterization}
Let $F:\mathcal P\to\mathcal R$ be a voting rule.

\begin{enumerate}[(i)]
\item
$F$ is the Baldwin rule if and only if it satisfies \textit{Neutrality}, \textit{Strong Bottom Consistency}, \textit{Weak Faithfulness}, \textit{Cancellation}, and \textit{Bottom Independence}.

\item
$F$ is the strict Nanson rule if and only if it satisfies \textit{Neutrality}, \textit{Bottom Consistency}, \textit{Survival Consistency}, \textit{Weak Faithfulness}, \textit{Dummy Retention}, and \textit{Bottom Independence}.

\item
$F$ is the Nanson rule if and only if it satisfies \textit{Neutrality}, \textit{Bottom Consistency}, \textit{Weak Faithfulness}, \textit{Dummy Exclusion}, and \textit{Bottom Independence}.
\end{enumerate}
\end{theorem}

The summary of the characterization of Theorems \ref{thm:recursive-borda} and \ref{thm:unified-characterization} is given in Table \ref{tab:axiom-comparison}.
The axioms in each characterization are independent;
See Appendix \ref{app:independence}.

\begin{table}[H]
\centering
\begin{tabular}{l|c|ccc}
\toprule
&
Regular recursive Borda
&
\multicolumn{3}{c}{Canonical rules}
\\
\cmidrule(lr){3-5}
Axiom
&
&
Baldwin
&
Strict Nanson
&
Nanson
\\
\midrule
Neutrality
&
\checkmark
&
\checkmark
&
\checkmark
&
\checkmark
\\
Bottom Consistency
&
\checkmark
&
&
\checkmark
&
\checkmark
\\
Strong Bottom Consistency
&
&
\checkmark
&
&
\\
Survival Consistency
&
&
&
\checkmark
&
\\
Weak Faithfulness
&
\checkmark
&
\checkmark
&
\checkmark
&
\checkmark
\\
Cancellation
&
\checkmark
&
\checkmark
&
&
\\
Dummy Retention
&
&
&
\checkmark
&
\\
Dummy Exclusion
&
&
&
&
\checkmark
\\
Bottom Independence
&
\checkmark
&
\checkmark
&
\checkmark
&
\checkmark
\\
\bottomrule
\end{tabular}
\caption{Comparison of the axiomatic foundations. The symbol $\checkmark$ indicates that the axiom is used in the corresponding characterization.}
\label{tab:axiom-comparison}
\end{table}

The theorem has a common first step.
In the Baldwin case, \textit{Strong Bottom Consistency} implies \textit{Bottom
Consistency}.
In the strict Nanson case, \textit{Dummy Retention} implies \textit{Cancellation}.
In the Nanson case, \textit{Dummy Exclusion} implies \textit{Cancellation}.
Therefore, each system of axioms satisfies the five axioms in
Theorem~\ref{thm:recursive-borda}.
Each rule consequently belongs to the identified class of recursive Borda rules in Theorem~\ref{thm:recursive-borda}.
The distinctive axioms of Proposition \ref{prop:canonical-local-identification} then determine the precise lower interval selected as the first bottom tier.

The three characterizations reveal that the Baldwin and Nanson rules represent two opposite approaches to recursive elimination, while the strict Nanson rule is not merely a technical intermediate procedure.

The Baldwin rule identifies the current bottom tier through \textit{Strong Bottom Consistency}, which uniquely determines the score-minimizing alternatives.
The Nanson rule instead identifies the bottom tier through \textit{Bottom Consistency} together with \textit{Dummy Exclusion}, thereby eliminating every alternative whose Borda score is weakly below the average.
The strict Nanson rule combines a different set of principles.
\textit{Bottom Consistency} preserves alternatives eliminated by both electorates, whereas \textit{Survival Consistency} preserves alternatives surviving in both electorates.
Thus, the rule is stable with respect to both elimination and survival.
At the same time, \textit{Dummy Retention} treats zero-score alternatives as belonging to the surviving side unless complete social indifference is reached.
Consequently, the distinction between the strict Nanson and Nanson rules is precisely whether the average Borda score belongs to the surviving region or the elimination region.
The three rules therefore correspond to three conceptually distinct recursive interpretations of Young's consistency principle rather than to minor procedural variants of the same elimination method.

\section{Discussion and concluding remarks}
\label{sec:discussion}
This paper develops a new class of voting rules, recursive Borda aggregation rules, and provides an axiomatic foundation for a regular subclass of them.
The class includes the Baldwin and the Nanson rules as two extreme rules with respect to the elimination process.
In particular, we show that, under Condorcet consistency together with an additional condition, only these two rules and the strict Nanson rule survive in the class.
Given that, we then provide axiomatic foundations of Baldwin, strict Nanson, and the Nanson rules in a unified way using comparable axioms of the Borda rule provided by \cite{Young1974}, thereby connecting Young's classical axiomatization of the Borda rule with a general theory of recursive social aggregation.
To conclude, we discuss several related issues and possible directions for future research.

\subsection{From the Borda rule to recursive Borda aggregation}
\label{subsec:one-shot-recursive-discussion}

Theorem~\ref{thm:recursive-borda} shows that the use of Borda scores in recursive aggregation need not be imposed directly.
\textit{Neutrality}, \textit{Bottom Consistency}, \textit{Weak Faithfulness}, and
\textit{Cancellation} imply that the current bottom tier depends only on the Borda-scores.
They also imply that this bottom correspondence is homogeneous and forms a lower interval of the Borda-score ordering bounded above by zero.
\textit{Bottom Independence} then turns this local bottom-tier decision into a recursive social ranking.
Thus, the Borda score is a derived property for the representation of the bottom tier rather than a primitive of the axioms.

This result creates a direct parallel with Young's characterization of the Borda rule.
The two branches share \textit{Neutrality} and \textit{Cancellation}, while \textit{Faithfulness} is weakened to \textit{Weak Faithfulness}.
They differ in the aggregation and continuation requirements.
Young's \textit{Consistency} preserves the complete social comparison across electorate aggregation and identifies the Borda rule, while \textit{Bottom Consistency} preserves only the current bottom-tier
decision.
\textit{Bottom Independence} then applies the resulting local criterion recursively to the surviving alternatives.

\subsection{Intermediate Condorcet winners and losers}
\label{subsec:intermediate-condorcet}

The three canonical rules also differ in their treatment of weak pairwise winners and
losers. Following \citet{BarberaBossert2025}, an alternative $x\in B$ is an
\emph{intermediate Condorcet winner} at $\succ\,\in\mathcal{P}(B)$ if
$[T(\succ)]_{xy}\ge 0$ for every $y\in B\setminus\{x\}$, with a strict inequality for at
least one such $y$. An alternative $x\in B$ is an \emph{intermediate Condorcet loser} if
$[T(\succ)]_{xy}\le 0$ for every $y\in B\setminus\{x\}$, with a strict inequality for at
least one such $y$.

The distinction rests on the fact that the sign of the centered Borda score of an
intermediate Condorcet winner is preserved under every restriction of the feasible set.
Indeed, if $x$ is an intermediate Condorcet winner at $\succ$, then
\[
  \beta^{C}_{x}\bigl(\succ\!\!|_{C}\bigr)
  \;=\;\sum_{y\in C\setminus\{x\}}[T(\succ)]_{xy}\;\ge\;0
\]
for every $C\subseteq B$ with $x\in C$. Fix such a $C$. If
$\beta^{C}(\succ\!\!|_{C})\neq \bm{0}_{C}$, then
$\min_{a\in C}\beta^{C}_{a}(\succ\!\!|_{C})<0\le \beta^{C}_{x}(\succ\!\!|_{C})$, so neither
the Baldwin rule nor the strict Nanson rule eliminates $x$ at $C$. If $\beta^{C}(\succ\!\!|_{C})=\bm{0}_{C}$, both rules stop and $C$ becomes the top tier.
Hence both rules place every intermediate Condorcet winner in the top tier.
Symmetrically, if $x$ is an intermediate Condorcet loser at $\succ$, then
$\beta^{B}_{x}(\succ)<0$, so both the strict Nanson rule and the Nanson rule eliminate
$x$ in the first round and place it in the bottom tier.

The Nanson rule does not share the winner-side protection, and the Baldwin rule does not
share the loser-side protection.
The strict Nanson rule consequently combines two properties that are separated by the Baldwin and Nanson rules.
It inherits the winner-side protection of the Baldwin rule and the loser-side protection of the Nanson rule.

\subsection{Local identification and recursive propagation}
\label{subsec:local-global-discussion}

The proofs separate local identification from recursive propagation.
\textit{Neutrality}, the relevant \textit{consistency} conditions, \textit{Weak Faithfulness}, and the boundary axioms determine the current bottom tier.
Then, \textit{Bottom Independence} determines how the rule continues after that tier has been removed.
Notice that \textit{Bottom Independence} does not itself identify a Borda-score region.
Its role is to convert a local characterization of the bottom tier into a characterization of the complete social ranking.

This separation suggests that weaker continuation conditions may be sufficient when the object of interest is a winner correspondence rather than a complete ranking.
While a complete-ranking characterization must determine every successive tier, a winner characterization may require only enough recursive structure to identify the final surviving alternatives.
Developing such correspondence versions of the present results is a natural direction for future research.

\subsection{Alternative proof methods}
\label{subsec:alternative-proofs}

The proofs in this paper use the linear-algebraic method initiated by \citet{Young1974}.
The cycle--cocycle decomposition identifies the centered Borda-score vector as the relevant representation of the current bottom tier.
The canonical-star argument then yields the lower-interval structure and the three canonical elimination regions.

A complementary approach works directly with preference-profile transformations.
\citet{HanssonSahlquist1976} use an amplification technique to obtain a direct proof of Young's characterization of the Borda rule.

For the Baldwin rule, \citet{GotoNakada2026note} provide a corresponding profile-based argument.
Whether a common amplification proof can cover the general recursive class together with the strict Nanson and Nanson characterizations remains open.
Such a proof would provide a profile-based counterpart to the unified
geometric argument developed here.


\begin{center}
\Large{{\bf Appendix}}
\end{center}

\numberwithin{definition}{section}
\numberwithin{theorem}{section}
\numberwithin{lemma}{section}
\numberwithin{proposition}{section}
\numberwithin{corollary}{section}
\numberwithin{example}{section}
\renewcommand{\theequation}{\thesection.\arabic{equation}}

\appendix

\section{Proof of Proposition \ref{prop:local-borda-structure}}\label{app:local-borda-structure}

This appendix proves Proposition \ref{prop:local-borda-structure}, which is a key step for Theorem \ref{thm:recursive-borda}.
The argument first reduces the bottom tier to a correspondence on weighted tournaments.
It then removes the cyclic component and identifies the resulting correspondence on centered Borda-score space.

\begin{lemma}[Bottom-tier margin reduction]
\label{lem:bottom-margin-reduction}

Suppose that $F$ satisfies \textit{Bottom Consistency} and \textit{Cancellation}.
Then, for every feasible set $B$ and every pair of profiles $\succ,\succ'\,\in\mathcal P(B)$, if  $T(\succ)=T(\succ')$, then $W_F(\succ)=W_F(\succ')$.
Moreover, for every profile $\succ$ and every positive integer $k$, $W_F(k*\!\succ)=W_F(\succ)$.
\end{lemma}

\begin{proof}
Fix profiles $\succ^N$ and $\succ^{N'}$ on the same feasible set $B$ such that $T(\succ^N)=T(\succ^{N'})$.
Suppose first that $N$ and $N'$ are disjoint.
Choose a profile $\widehat\succ$ on an electorate satisfying $T(\widehat\succ)=T(\succ^N)$.
The matrix $-2T(\succ^N)$ is an even integer-valued skew-symmetric matrix.
Hence, by the majority-margin realizability theorem of
\citet{Debord1987}, there exists a profile $\succ^\ast$ on another electorate satisfying $T(\succ^\ast)=-2T(\succ^N)$.
All four electorates may be chosen to be pairwise disjoint.

Define $\zeta=\; \succ^N \! \cup \;\widehat\succ \;\cup\succ^\ast$ and $\zeta'=\; \succ^{N'} \! \cup \; \widehat\succ \; \cup\succ^\ast$.
Both profiles have zero majority-margin matrices.
\textit{Cancellation} gives $F(\zeta)=F(\zeta')=(B)$.

Now define
\[
\eta = \; \succ^N \! \cup \succ^{N'} \!
\cup \; \widehat\succ \; \cup \succ^\ast \!.
\]
The profile $\eta$ can be written as $\succ^N \! \cup \; \zeta'$ and as $\succ^{N'}\! \cup \; \zeta$.
Then, (ii) of \textit{Bottom Consistency} therefore gives $W_F(\succ^N)=W_F(\eta)=W_F(\succ^{N'})$.

If the original electorates overlap, choose disjoint copies of the two profiles.
The disjoint-electorate argument first identifies each original profile with its copy and then identifies the two copies.
Thus, the bottom tier depends only on the majority-margin matrix.

We next prove homogeneity of $W_F$.
Fix $\succ\,\in\mathcal P(B)$.
We proceed by induction on $k$.
The result is immediate for $k=1$.

Suppose that $W_F(k* \! \succ)=W_F(\succ)$.
Let $\widehat\succ$ be a disjoint copy of $\succ$.
Margin dependence of $W_F$ proved above gives
\[
W_F(\widehat\succ)=W_F(\succ).
\]
Then, (iii) of \textit{Bottom Consistency} therefore yields
\[
W_F((k+1)*\! \succ) = W_F(\succ).
\]
The result follows by induction.

\end{proof}

We then extend the bottom correspondence to weighted tournaments.
For every feasible set $B$, let $\mathcal V(B)=\left\{T\in\mathbb Q^{B\times B}:T^\top=-T\right\}$ denote the vector space of rational weighted tournaments on $B$.

\begin{lemma}
\label{lem:bottom-tournament-extension}
Suppose that $F$ satisfies the conclusions of
Lemma~\ref{lem:bottom-margin-reduction}.
Then, for every feasible set $B$, there exists a unique correspondence
$\Omega_B:\mathcal V(B)\longrightarrow 2^{B}\setminus\{\emptyset\}$ such that
\[
  \Omega_B(\theta)=W_F(\succ)
  \quad
  \text{whenever } \theta\in\mathcal V(B),\ \succ\,\in\mathcal P(B),
  \text{ and } k\in\mathbb N \text{ satisfy } T(\succ)=k\theta .
\]
In particular, $\Omega_B(T(\succ))=W_F(\succ)$ for every $\succ\,\in\mathcal P(B)$.
Moreover, $\Omega_B(\lambda\theta)=\Omega_B(\theta)$ for every positive rational number $\lambda$.
\end{lemma}


\begin{proof}
Take any $\theta\in\mathcal V(B)$ and a positive integer $k$ such that every off-diagonal entry of $k\theta$ is an even integer.
By \citet{Debord1987}, there exists a profile $\succ$ satisfying $T(\succ)=k\theta$.

Suppose that another profile $\succ'$ and another positive integer $k'$ satisfy $T(\succ')=k'\theta$.
Then, $T(k'*\!\succ)=kk'\theta=T(k*\!\succ')$.
By Lemma~\ref{lem:bottom-margin-reduction}, 
$W_F(\succ)=W_F(k'*\!\succ)=W_F(k*\!\succ')=W_F(\succ')$, which implies that $\Omega_B$ is well defined.
The same argument gives positive rational homogeneity, that is, $\Omega_B(\lambda \theta)=\Omega_B(\theta)$ for every positive rational number $\lambda$.
\end{proof}

The next step is to reduce the domain $\mathcal{V}(B)$ to the space of centered Borda-scores.
For $T\in\mathcal V(B)$, define its row-sum vector by $\beta(T)=\left(\sum_{y\in B\setminus\{x\}}T_{xy}\right)_{x\in B}$.
For $s\in\mathcal S(B)$ and $q=|B|$, define
$K_B(s)\in\mathcal V(B)$ by
\[
[K_B(s)]_{xy}
=
\frac{s_x-s_y}{q}.
\]
A direct calculation gives $\beta(K_B(s))=s$.
For three distinct alternatives $x,y,z\in B$, let $C^{xyz}$ denote the unit three-cycle tournament.
We use the standard cycle--cocycle decomposition of weighted tournaments
\citep[Mathematical Comments 1, 2 and 4]{Zwicker1991}; see also
\citet{BrandlPeters2019} and \citet{Zwicker2018}.

\begin{lemma}[Cycle--cocycle decomposition]\label{lem:cycle-cocycle}
For every $T\in\mathcal{V}(B)$ there is a unique $C\in\mathcal{V}(B)$ with
$\beta(C)=\bm{0}_{B}$ such that $T=K_{B}(\beta(T))+C$. If $|B|\ge 3$, the subspace
$\{C\in\mathcal{V}(B):\beta(C)=\bm{0}_{B}\}$ is spanned by the unit three-cycle
tournaments.
\end{lemma}



Under \textit{Neutrality}, \textit{Bottom Consistency}, and \textit{Cancellation}, the corresponding properties are inherited by $\Omega$.
In particular,
\[
\Omega_{\sigma(B)}(\sigma(T))
=
\sigma\bigl(\Omega_B(T)\bigr),
\]
and electorate union is replaced by matrix addition.
By Lemma~\ref{lem:cycle-cocycle}, we can show that the mapping $\Omega$ satisfies a variant of \textit{Cancellation} in the following sense.

\begin{lemma}[Cyclic indifference]
\label{lem:cyclic-indifference}
Suppose that $F$ satisfies \textit{Neutrality}, \textit{Bottom Consistency}, and \textit{Cancellation}.
Then, $\Omega_B(C)=B$ for every $C\in\mathcal V(B)$ satisfying $\beta(C)=\bm 0_B$.
\end{lemma}

\begin{proof}
The conclusion is immediate when $|B|\leq2$.
Suppose that $|B|\geq3$.
Consider first a unit three-cycle $C^{xyz}$.
Neutrality under cyclic permutations of $x,y,z$ implies that either all three cycle alternatives or none of them belong to the bottom tier.
Also, Neutrality under permutations of the alternatives outside the cycle implies the same symmetry among all outside alternatives.

Let $\sigma$ exchange $x$ and $y$: $\sigma(x)=y$ and $\sigma(y)=x$.
The transposition leaves each of the two symmetry classes invariant and satisfies $\sigma(C^{xyz})=-C^{xyz}$.
\textit{Neutrality} therefore gives $\Omega_B(-C^{xyz})=\Omega_B(C^{xyz})$.
If this common bottom tier were a proper subset of $B$, (iii) of  \textit{Bottom Consistency} would give the same proper bottom tier at
\[
C^{xyz}+(-C^{xyz})=0,
\]
which contradicts \textit{Cancellation}.
Hence, $\Omega_B(C^{xyz})=B$.

Now let $C$ be an arbitrary zero-row-sum tournament.
By Lemma~\ref{lem:cycle-cocycle}, it is a rational linear combination of unit three-cycle tournaments.
After reversing the orientation of cycles with negative coefficients, it can be written as a nonnegative rational combination of such
tournaments.
Positive homogeneity and repeated applications of part (ii) of \textit{Bottom Consistency} imply $\Omega_B(C)=B$.
\end{proof}

Define $\Delta_B(s)=\Omega_B(K_B(s))$ for every $s\in\mathcal S(B)$.
The correspondence $\Delta_B$ is neutral and positively homogeneous, so that it satisfies $\Delta_B(\bm 0_B)=B$.
In addition, Lemma~\ref{lem:cycle-cocycle} directly implies that the bottom $W_F$ depends only on the Borda score.

\begin{lemma}
\label{lem:bottom-tier-borda-dependence}
For every $T\in\mathcal V(B)$,
\[
\Omega_B(T)
=
\Delta_B(\beta(T)).
\]
Consequently, for every profile $\succ\,\in\mathcal P(B)$,
\[
W_F(\succ)
=
\Delta_B\bigl(\beta^B(\succ)\bigr).
\]

\end{lemma}

\begin{proof}
By Lemma~\ref{lem:cycle-cocycle}, write $T=K_B(\beta(T))+C$, where $\beta(C)=\bm 0_B$.
Moreover, cyclic indifference gives $\Omega_B(C)=B$.
Then, (ii) of \textit{Bottom Consistency} therefore implies $\Omega_B(T)=\Omega_B(K_B(\beta(T)))$.
The conclusion follows from the definition of $\Delta_B$.
\end{proof}

Combining Lemmas~\ref{lem:bottom-tournament-extension} and \ref{lem:bottom-tier-borda-dependence}, we have shown (i)--(iii) of Proposition \ref{prop:local-borda-structure}.

Next, we show (iv)--(v) of Proposition \ref{prop:local-borda-structure}.
By definition of $\Delta_B$, the following observation immediately holds by repeated application of \textit{Bottom Consistency}.

\begin{lemma}[Preservation of common bottom alternatives]
\label{lem:common-bottom-preservation}
Let $s^1,\ldots,s^\ell\in\mathcal S(B)$.
If $x\in\Delta_B(s^j)$ for every $j \in \{ 1, \dots, \ell\}$, then
\[
x\in
\Delta_B
\left(
\sum_{j=1}^{\ell}s^j
\right).
\]
Moreover, if
$\Delta_B(s^j)=B$
for every $j$, then
\[
\Delta_B
\left( \sum_{j=1}^{\ell}s^j \right) = B.
\]
\end{lemma}

Let us consider the following specific vectors referred to as \textit{star vectors}.
Fix a feasible set $B$ with $q=|B|\geq2$.
For every $x\in B$, define $u_B^x\in\mathcal S(B)$ by
\[
(u_B^x)_y=
\begin{cases}
q-1~\text{if}~y=x,\\
-1~\text{if}~y \neq x.
\end{cases}
\]

\begin{lemma}
\label{lem:canonical-stars}
Suppose that $F$ satisfies \textit{Neutrality}, \textit{Bottom Consistency}, \textit{Weak Faithfulness}, and \textit{Cancellation}.
For every feasible set $B$ with $|B|\geq2$ and every $x\in B$, $\Delta_B(u_B^x)=B\setminus\{x\}$ and $\Delta_B(-u_B^x)=\{x\}$.
\end{lemma}

\begin{proof}
The vector $u_B^x$ is invariant under every permutation of $B\setminus\{x\}$.
\textit{Neutrality} therefore implies that $\Delta_B(u_B^x)$ is one of $\{x\}, B\setminus\{x\}, B$.
The same conclusion holds for $\Delta_B(-u_B^x)$.
Since $u_B^x+(-u_B^x)=\bm 0_B$, the two bottom tiers cannot be the same proper subset.
Otherwise, (iii) of \textit{Bottom Consistency} would contradict $\Delta_B(\bm 0_B)=B$.

If one of the two bottom tiers is $B$, (ii) of \textit{Bottom Consistency} implies that the other is also $B$.
Suppose that both bottom tiers are $B$.
\textit{Neutrality} gives the same conclusion for every alternative in place of
$x$.
Let $v$ be the centered Borda-score vector of a one-voter profile.
Then, we can see that $v=(1/q)\sum_{y\in B}v_yu_B^y$ holds.
Replace each term with a negative coefficient by a positive multiple of $-u_B^y$.
Every component then has bottom tier $B$.
Lemma~\ref{lem:common-bottom-preservation} gives $\Delta_B(v)=B$, which contradicts \textit{Weak Faithfulness}.
Hence, neither star vector has bottom tier $B$.

Suppose, contrary to the claim, that $\Delta_B(u_B^x)=\{x\}$ and $\Delta_B(-u_B^x)=B\setminus\{x\}$.
\textit{Neutrality} gives the analogous assignment for every alternative.
Let $v$ be the centered Borda-score vector of a one-voter profile
$\succ_i$ in which $x$ is top-ranked. We then have $v_x>v_y$ for every $y\neq x$.

Since $\sum_{y\in B}v_y=0$,
\begin{equation}\label{eq:delta}
  \sum_{y\in B\setminus\{x\}}(v_x-v_y)
  =(q-1)v_x-\sum_{y\in B\setminus\{x\}}v_y
  =(q-1)v_x+v_x
  =qv_x .
\end{equation}
Fixing $z\in B$, we have
$(-u_B^{y})_z=1-q\,\delta_{yz}$.
Hence, by (\ref{eq:delta}),
\[
  \Biggl[\;\sum_{y\in B\setminus\{x\}}(v_x-v_y)\bigl(-u_B^{y}\bigr)\Biggr]_z
  =\sum_{y\in B\setminus\{x\}}(v_x-v_y)
   -q(v_x-v_z)\,\delta_{z\neq x}
  =qv_x-q(v_x-v_z)\,\delta_{z\neq x}
  =qv_z ,
\]
where the last equality holds both for $z=x$ and for $z\neq x$.
Therefore
\[
  v=\frac{1}{q}\sum_{y\in B\setminus\{x\}}(v_x-v_y)\bigl(-u_B^{y}\bigr),
\]
and every coefficient $(v_x-v_y)/q$ is strictly positive.
For every $y\neq x$, the set $\Delta_B(-u_B^y)=B\setminus\{y\}$ contains $x$.
Lemma~\ref{lem:common-bottom-preservation} therefore gives $x\in\Delta_B(v)$, which contradicts \textit{Weak Faithfulness}.
Thus, $\Delta_B(u_B^x)=B\setminus\{x\}$ and
$\Delta_B(-u_B^x)=\{x\}$.
\end{proof}

\begin{lemma}\label{lem:positive-survive}
For every $s \in \mathcal{S}(B)$ and every $x \in B$, if $s_x > 0$, then
$x \notin \Delta_B(s)$.
\end{lemma}

\begin{proof}
Suppose, contrary to the claim, that $x \in \Delta_B(s)$. Let $\Pi_x(B)$
denote the group of permutations of $B$ fixing $x$. \textit{Neutrality} gives
$x \in \Delta_B(\sigma(s))$ for every $\sigma \in \Pi_x(B)$. Lemma~\ref{lem:common-bottom-preservation}
therefore gives
\[
  x \in \Delta_B
  \Biggl( \sum_{\sigma \in \Pi_x(B)} \sigma(s) \Biggr).
\]
Let $q = |B|$.
By symmetry,
\[
  \sum_{\sigma \in \Pi_x(B)} \sigma(s) = (q-2)!\, s_x\, u^x_B .
\]
Since $s_x > 0$, positive homogeneity and Lemma \ref{lem:canonical-stars} imply that the bottom tier of this vector is $B \setminus \{x\}$, which is a contradiction.
\end{proof}

\begin{corollary}[Nondegeneracy]\label{cor:nondegeneracy}
For every $s \in \mathcal{S}(B)$, $\Delta_B(s) = B$ if and only if $s = 0_B$.
\end{corollary}

\begin{proof}
Cancellation gives $\Delta_B(0_B) = B$. Conversely, suppose that
$s \neq 0_B$. Since $\sum_{a \in B} s_a = 0$, there exists $x \in B$ with
$s_x > 0$. Lemma~\ref{lem:positive-survive} gives $x \notin \Delta_B(s)$,
and hence $\Delta_B(s) \neq B$.
\end{proof}

We now show (v) of Proposition \ref{prop:local-borda-structure}.
\begin{lemma}
\label{lem:derived-lower-interval}
For every $s\in\mathcal S(B)$ and every $a,b\in B$, if $a\in\Delta_B(s)$ and $s_b\leq s_a$, then $b\in\Delta_B(s)$.
\end{lemma}

\begin{proof}
Fix $a,b\in B$ satisfying the assumptions, and let $q=|B|$.
The claim is trivial when $a=b$, so assume $a\neq b$.
Let $\sigma_{ab}\in\Pi(A)$ be the transposition of $a$ and $b$.
Since $\sigma_{ab}(B)=B$ and $a\in\Delta_B(s)$, Neutrality gives
$b=\sigma_{ab}(a)\in\Delta_B\bigl(\sigma_{ab}(s)\bigr)$.
If $s_a=s_b$, then $\sigma_{ab}(s)=s$ and the conclusion follows.

Suppose that $s_a>s_b$ and set $d=s_a-s_b>0$.
Define $h\in\mathcal S(B)$ by
\[
  h_a=d,\qquad h_b=-d,\qquad h_z=0
  \quad\text{for every } z\in B\setminus\{a,b\}.
\]
It can be checked here that
\[
  h=\frac{d}{q}\,u_B^{a}+\frac{d}{q}\bigl(-u_B^{b}\bigr).
\]
Since $\Delta_B(u_B^{a})=B\setminus\{a\}\ni b$ and $\Delta_B(-u_B^{b})=\{b\}$, positive homogeneity and Lemma~\ref{lem:common-bottom-preservation} gives $b\in\Delta_B(h)$.
Then, we have $\sigma_{ab}(s)+h=s$, and since the two vectors agree with $s$ off
$\{a,b\}$,
Lemma~\ref{lem:common-bottom-preservation} gives $b\in\Delta_B(s)$.
\end{proof}



\begin{proof}[Proof of Proposition \ref{prop:local-borda-structure}]
By Lemmas \ref{lem:bottom-margin-reduction}--\ref{lem:derived-lower-interval}, we have shown the properties (i)--(v).
It remains to prove the condition (vi). 
Fix $s,t\in S(B)$.
Choose a positive integer $k$ such that both $ks$ and $kt$ are centered Borda-score vectors of profiles on pairwise disjoint electorates. More precisely, by the majority-margin realizability theorem of \citet{Debord1987}, there exist profiles $\succ,\succ'\,\in P(B)$ on disjoint electorates such that $\beta^B(\succ)=ks$ and $\beta^B(\succ')=kt$.
Positive homogeneity, (iii), gives $W_F(\succ)=\Delta_B(s)$, $W_F(\succ')=\Delta_B(t)$, and $W_F(\succ\cup\succ')=\Delta_B(s+t)$.
If $\Delta_B(s)\cap\Delta_B(t)\neq\emptyset$, part (i) of \textit{Bottom Consistency} implies
\[
\Delta_B(s)\cap\Delta_B(t) \subseteq \Delta_B(s+t).
\]
If, in addition, $\Delta_B(s)=\Delta_B(t)$, part (iii) of \textit{Bottom Consistency} gives
\[
\Delta_B(s+t)=\Delta_B(s)=\Delta_B(t).
\]
Hence, condition (vi) holds.
\end{proof}

\section{Proof of Proposition \ref{prop:canonical-local-identification}}
\label{app:canonical-local-identification}

We begin with the following mathematical observation on the existence of a weighted tournament with a prescribed score vector $s\in\mathcal S(B)$.

\begin{lemma}
\label{lem:coordinate-completion}
Let $B$ be a feasible set with $q=|B|\geq2$ and $s\in\mathcal S(B)$.
For any $x\in B$, there exists a rational weighted tournament $T$ satisfying
\[
\beta(T)=s \qquad \text{and} \qquad
T_{xy}=\frac{s_x}{q-1}
\]
for every $y\neq x$.
In particular, if $s_x=0$, the tournament can be chosen so that $x$ is a dummy alternative.

\end{lemma}

\begin{proof}
Set
$r=\frac{s_x}{q-1}$.
For every $y\neq x$, define $T_{xy}=r$ and $T_{yx}=-r$.
Let $C=B\setminus\{x\}$ and define $d_y=s_y+r$ for every $y\in C$.
Since $\sum_{y\in C}d_y= \sum_{y \in C}s_y + s_x = 0$,
define
\[
T_{yz}
=
\frac{d_y-d_z}{q-1}
\]
for every $y,z\in C$.
By construction, the matrix is skew-symmetric.
The row sum of $y\in C$ within $C$ is $d_y$, so that its full row sum is therefore $d_y-r=s_y$.
Similarly, the row sum of $x$ is
$(q-1)r=s_x$.
Thus, $\beta(T)=s$.
If $s_x=0$, then $r=0$ and $x$ is a dummy.
\end{proof}

Recall from Lemma~\ref{lem:positive-survive} that no alternative with a positive centered Borda score belongs to the bottom tier.
By applying \textit{Strong Bottom Consistency}, we can obtain the following.

\begin{lemma}
\label{lem:iterated-bottom-consistency}
Suppose that $\Delta_B$ satisfies \textit{Strong Bottom Consistency}.
Let $s^1,\ldots,s^\ell\in\mathcal S(B)$ satisfy
\[
\bigcap_{j=1}^{\ell}\Delta_B(s^j)
\neq\emptyset.
\]
Then
\[
\Delta_B
\left(
\sum_{j=1}^{\ell}s^j
\right)
=
\bigcap_{j=1}^{\ell}\Delta_B(s^j).
\]
\end{lemma}
\begin{proof}
The result follows by induction on $\ell$.
Every partial intersection contains the nonempty full intersection.
\textit{Strong Bottom Consistency} is therefore applicable at each step.
\end{proof}

We are ready to prove Proposition~\ref{prop:canonical-local-identification}.

\begin{proof}[Proof of
Proposition~\ref{prop:canonical-local-identification}]

We prove the three parts separately.

\medskip
\noindent
\paragraph{(i) The Baldwin rule.}
Fix $s\in\mathcal S(B)$ and let $\underline s=\min_{y\in B}s_y$.
Define
\[
c_y
=
\frac{s_y-\underline s}{|B|}
\]
for every $y\in B$.
Then,  $c_y\geq0$ and $s=\sum_{y\in B}c_yu_B^y$.
The coefficient $c_y$ is positive exactly when $s_y>\underline s$.
For every $y$ with $c_y>0$, Lemma \ref{lem:canonical-stars} gives
\[
\Delta_B(c_yu_B^y)
=
B\setminus\{y\}.
\]
Hence, the intersection of these bottom tiers is
\[
\arg\min_{y\in B}s_y,
\]
which is nonempty.
Therefore, Lemma \ref{lem:iterated-bottom-consistency} gives
\[
\Delta_B(s)=\arg\min_{y\in B}s_y.
\]

\medskip
\noindent
\paragraph{(ii) The strict Nanson rule.}
If $s=\bm 0_B$, \textit{Cancellation} gives
$\Delta_B(s)=B$.

Suppose that $s\neq\bm 0_B$ and fix $x\in B$.
If $s_x>0$, Lemma~\ref{lem:positive-survive} gives $x\notin\Delta_B(s)$.
So, suppose that $s_x=0$.
By Lemma~\ref{lem:coordinate-completion}, the same score vector is induced by a tournament in which $x$ is a dummy.
If $x\in\Delta_B(s)$, \textit{Dummy Retention} implies complete social indifference, which contradicts $\Delta_B(s)\neq B$.
Hence, $x\notin\Delta_B(s)$.

Suppose finally that $s_x<0$.
Define
\[
\alpha
=
\frac{-s_x}{|B|-1},
\qquad
g
=
\alpha(-u_B^x),
\qquad
h
=
s-g.
\]
Then
$h_x=0$.
If $h=\bm 0_B$, Lemma \ref{lem:canonical-stars} gives $x\in\Delta_B(s)$.
Suppose that $h\neq\bm 0_B$.
The zero-coordinate argument applied to $-h$ gives $x\notin\Delta_B(-h)$.
Assume, contrary to the claim, that
$x\notin\Delta_B(s)$.
The alternative $x$ survives at both $s$ and $-h$.
Moreover,$s+(-h)=g$ and $\Delta_B(g)=\{x\}\neq B$.
\textit{Survival Consistency} implies that $x$ survives at $g$, which contradicts $\Delta_B(g)=\{x\}$.
Hence, $x\in\Delta_B(s)$.

Summarizing all the arguments, 
\[
\Delta_B(s)
=
\{x\in B:s_x<0\}
\]
for every nonzero $s$.

\medskip
\noindent
\paragraph{(iii) The Nanson rule.}
Fix $s\in\mathcal S(B)$ and $x\in B$.
If $s_x>0$, Lemma~\ref{lem:positive-survive} gives $x\notin\Delta_B(s)$.
Suppose that $s_x=0$.
By Lemma \ref{lem:coordinate-completion}, we can find a weighted tournament with score vector $s$ in
which $x$ is a dummy.
\textit{Dummy Exclusion} gives $x\in\Delta_B(s)$.

Suppose that
$s_x<0$.
Define
\[
\alpha
=
\frac{-s_x}{|B|-1},
\qquad
g
=
\alpha(-u_B^x),
\qquad
h
=
s-g.
\]

Then $h_x=0$.
Lemma \ref{lem:canonical-stars} gives $x\in\Delta_B(g)$.
The zero-coordinate argument gives $x\in\Delta_B(h)$.
\textit{Bottom Consistency} therefore gives
\[
x\in\Delta_B(s).
\]
Consequently,
\[
\Delta_B(s)
=
\{x\in B:s_x\leq0\}.
\]
\end{proof}

\section{Logical independence}
\label{app:independence}
We consider the \textit{weighted-margin rules} as follows.
Fix strictly positive symmetric weights $\lambda_{\{x,y\}}=\lambda_{\{y,x\}}$ and define $q_x^{\lambda,B}(\succ)=\sum_{y\in B\setminus\{x\}}\lambda_{\{x,y\}}[T(\succ)]_{xy}$.
Let $F^{\lambda,\min}$ eliminate the minimizers of $q^{\lambda,B}$.
Let $F^{\lambda,<}$ eliminate the alternatives with $q_x^{\lambda,B}<0$.
Let $F^{\lambda,\leq}$ eliminate the alternatives with $q_x^{\lambda,B}\leq0$.

Let $F^{\mathrm{CW}}$ rank a Condorcet winner uniquely first whenever one exists and otherwise yield complete indifference.
Let $F^{\mathrm{Pl}}$ recursively eliminate the alternatives with the lowest plurality score.
Let $F^{\max}$, $F^{>}$, and $F^{\geq}$ recursively eliminate the alternatives with, respectively, the highest, positive, and nonnegative centered Borda scores until $s=\bm0_{B}$.
Let $F^{\mathrm{loss}}$ recursively eliminate
$L_{B}(\succ)=\{x\in B:[T(\succ)]_{xy}<0\text{ for some }y\in B\setminus\{x\}\}$,
with the convention that $L_{B}(\succ)=B$ whenever $L_{B}(\succ)=\emptyset$.

For $G \in\{F^{\mathrm{Baldwin}},F^{\mathrm{sN}},F^{\mathrm{Nanson}}\}$, let $\widehat F^G$ have the same bottom tier as $G$ and tie all alternatives outside that tier.
The independence of the axioms for each rule is summarized in Table \ref{tab:independence-witnesses}.

\begin{table}[H]
\begin{center}
\begin{tabular}{lll}
\toprule
Rule & Omitted axiom & Counterexample \\
\midrule
Baldwin & Neutrality & $F^{\lambda,\min}$ \\
& Strong Bottom Consistency & $F^{\mathrm{CW}}$ \\
& Weak Faithfulness & $F^{\max}$ \\
& Cancellation & $F^{\mathrm{Pl}}$ \\
& Bottom Independence & $\widehat F^{\mathrm{Baldwin}}$ \\
\midrule
Strict Nanson & Neutrality & $F^{\lambda,<}$ \\
& Bottom Consistency & $F^{\mathrm{loss}}$ \\
& Survival Consistency & $F^{\mathrm{Baldwin}}$ \\
& Weak Faithfulness & $F^{>}$ \\
& Dummy Retention & $F^{\mathrm{Nanson}}$ \\
& Bottom Independence & $\widehat F^{\mathrm{sN}}$ \\
\midrule
Nanson & Neutrality & $F^{\lambda,\leq}$ \\
& Bottom Consistency & $F^{\mathrm{CW}}$ \\
& Weak Faithfulness & $F^{\geq}$ \\
& Dummy Exclusion & $F^{\mathrm{sN}}$ \\
& Bottom Independence & $\widehat F^{\mathrm{Nanson}}$ \\
\bottomrule
\end{tabular}
\caption{Independence of axioms}
\label{tab:independence-witnesses}
\end{center}
\end{table}

We verify that each rule satisfies the claim.
For three alternatives $B=\{a,b,c\}$, write $T[u,v,w]$ for the weighted tournament satisfying $T_{ab}=u$, $T_{ac}=v$, and $T_{bc}=w$.

\paragraph{Weighted-margin rules.}
The weighted-margin vector is additive and has coordinate sum zero.
Consequently, $F^{\lambda,\min}$ satisfies \textit{Strong Bottom Consistency}, while $F^{\lambda,<}$ and $F^{\lambda,\leq}$ satisfy \textit{Bottom Consistency}.
The rule $F^{\lambda,<}$ also satisfies \textit{Survival Consistency}.
Strict positivity of the weights implies \textit{Weak Faithfulness}.
A dummy alternative has weighted-margin score zero, so $F^{\lambda,<}$ satisfies \textit{Dummy Retention} and $F^{\lambda,\leq}$ satisfies \textit{Dummy Exclusion}.
All three weighted-margin rules satisfy \textit{Cancellation} and \textit{Bottom Independence}.

They may violate \textit{Neutrality} when the weights are not invariant under permutations.
For example, let
\[
\lambda_{\{a,b\}}=2
\qquad\text{and}\qquad
\lambda_{\{a,c\}}=\lambda_{\{b,c\}}=1,
\]
and consider the permutation $\sigma$ defined by $\sigma(a)=c$, $\sigma(c)=b$, and $\sigma(b)=a$.
At $T[-4,-2,4]$, the weighted-margin vector is $(-10,12,-2)$, whereas at the permuted tournament it is $(12,-6,-6)$.
Thus, the minimum-score set is not preserved by $\sigma$, so $F^{\lambda,\min}$ violates \textit{Neutrality}.
At $T[-6,-6,-6]$, the corresponding vectors are $(-18,6,12)$ and $(-6,18,-12)$.
The negative and nonpositive sets are not preserved by $\sigma$, so $F^{\lambda,<}$ and $F^{\lambda,\leq}$ also violate \textit{Neutrality}.

\paragraph{Condorcet winner first--complete indifference for others, $F^{\mathrm{CW}}$.}
The rule $F^{\mathrm{CW}}$ satisfies \textit{Neutrality}, \textit{Weak Faithfulness}, \textit{Cancellation}, \textit{Dummy Exclusion}, and \textit{Bottom Independence}.
Consider
\[
T=T[4,-4,4]
\qquad\text{and}\qquad
T'=T[-2,8,-2].
\]
Neither tournament has a Condorcet winner, so both outcomes are complete indifference.
Their sum is $T[2,4,2]$, at which $a$ is a Condorcet winner.
Hence, the bottom tier changes from $B$ in both components to $\{b,c\}$ in their sum.
Thus, $F^{\mathrm{CW}}$ violates both \textit{Strong Bottom Consistency} and \textit{Bottom Consistency}.

\paragraph{The recursive plurality-loser rule, $F^{Pl}$.}
This rule is the set-valued variant of instant-runoff voting, also known
as the alternative vote or Hare's rule, and referred to as STV in the
computational social choice literature.\footnote{Unlike the standard formulation, $F^{Pl}$ simultaneously eliminates every alternative attaining the lowest plurality score, so that no tie-breaking is required. \citet{FreemanBrillConitzer2014} instead resolve such ties through the parallel-universes framework.}
The rule $F^{Pl}$ satisfies \textit{Neutrality}, \textit{Strong Bottom Consistency}, \textit{Weak Faithfulness}, and \textit{Bottom Independence}.
It violates \textit{Cancellation}.
Indeed, consider the profile
\[
\begin{array}{c|cccc}
\text{Multiplicity} & 2 & 2 & 1 & 1 \\
\hline
\text{Ranking}
& a\succ b\succ c
& c\succ b\succ a
& b\succ a\succ c
& c\succ a\succ b
\end{array}
\]
Every ranking is paired with its reverse, so the majority-margin matrix is zero.
The plurality scores are $(2,1,3)$, and the rule eliminates $b$.

\paragraph{The score-reversal rules, $F^{\max}$, $F^{>}$, and $F^{\geq}$.}
Each rule satisfies the relevant axioms of the Baldwin, the strict Nanson, and Nanson rule, respectively, except for \textit{Weak Faithfulness}.
This is because the top alternative in a one-voter profile has the highest positive centered Borda score and is placed in the bottom tier.

\paragraph{Recursive elimination, $F^{\mathrm{loss}}$.}
The rule $F^{\mathrm{loss}}$ satisfies \textit{Neutrality}, \textit{Survival Consistency}, \textit{Weak Faithfulness}, \textit{Dummy Retention}, and \textit{Bottom Independence}.
The rule violates \textit{Bottom Consistency}.
At
\[
T=T[-4,-4,-4]
\qquad\text{and}\qquad
T'=T[2,0,4],
\]
the bottom tiers are $\{a,b\}$ and $\{b,c\}$.
Their sum is $T[-2,-4,0]$, whose bottom tier is $\{a\}$.
Thus, the common bottom alternative $b$ is not preserved.

\paragraph{The Baldwin rule.}
The Baldwin rule satisfies every axiom in the strict Nanson characterization except \textit{Survival Consistency}.
For score vectors
\[
s=(-2,-4,6)
\qquad\text{and}\qquad
t=(-2,6,-4),
\]
the alternative $a$ survives at both vectors.
At $s+t=(-4,2,2)$, however, $a$ is the unique bottom alternative.

\paragraph{The Nanson and strict Nanson rules.}
The Nanson rule satisfies every axiom in the strict Nanson characterization except \textit{Dummy Retention}.
At $T[0,0,2]$, the centered Borda-score vector is $(0,2,-2)$ and $a$ is a dummy.
The Nanson bottom tier is $\{a,c\}$, although the outcome is not complete indifference.

The strict Nanson rule satisfies every axiom in the Nanson characterization except \textit{Dummy Exclusion}.
At the same tournament, its bottom tier is $\{c\}$, so the dummy alternative $a$ survives.

\paragraph{Truncation rules, $\widehat F^G$ for $G\in\{F^{\mathrm{Baldwin}},F^{\mathrm{sN}},F^{\mathrm{Nanson}}\}$.}
Each truncation $\widehat F^G$ has the same bottom tier as its underlying canonical rule.
It therefore satisfies every axiom concerning only the bottom tier, but violates \textit{Bottom Independence}.
At $T[-4,-2,-2]$, the centered Borda-score vector is $(-6,2,4)$, and all three canonical rules have the bottom tier $\{a\}$.
On $\{b,c\}$, the margin $T_{bc}=-2$ ranks $c$ strictly above $b$.
The truncation instead ties $b$ and $c$.

\bibliography{ref.bib}

\end{document}